\documentclass[graybox]{svmult}

\usepackage{type1cm}        
\usepackage{makeidx}         
\usepackage{graphicx}        
\usepackage{multicol}        
\usepackage[bottom]{footmisc}

\usepackage{newtxtext}       %
\usepackage[varvw]{newtxmath}       

\makeindex             

\usepackage{multirow,tikz-cd}
\allowdisplaybreaks
\theoremstyle{definition}
\newtheorem{Thm}{Theorem}[section]
\newtheorem{exmp}[Thm]{Example}
\newtheorem{Def}[Thm]{Definition}
\newtheorem{Coroll}[Thm]{Corollary}

\newtheorem{Prop}[Thm]{Proposition}
\begin{document}

\title*{Multi-twisted codes as free modules over principal ideal domains}
\titlerunning{MT codes as free modules over PIDs}
\author{Ramy Taki ElDin}
\institute{Ramy Taki ElDin \at Faculty of Engineering-Ain Shams University, Cairo, Egypt, \email{ramy.farouk@eng.asu.edu.eg}}
%
%
\maketitle

\abstract{We begin this chapter by introducing the simple algebraic structure of cyclic codes over finite fields. This structure undergoes a series of generalizations to present algebraic descriptions of constacyclic, quasi-cyclic (QC), quasi-twisted (QT), generalized quasi-cyclic (GQC), and multi-twisted (MT) codes. The correspondence between these codes and submodules of the free $\mathbb{F}_q[x]$-module $\left(\mathbb{F}_q[x]\right)^\ell$ is established. Thus, any of these codes corresponds to a free linear code over the principal ideal domain (PID) $\mathbb{F}_q[x]$. A basis of this code exists and is used to build a generator matrix with polynomial entries, called the generator polynomial matrix (GPM). The Hermite normal form of matrices over PIDs is exploited to achieve the reduced GPMs of MT codes. Some properties of the reduced GPM are introduced, for example, the identical equation. A formula for a GPM of the dual code $\mathcal{C}^\perp$ of a MT code is given. At this point, special attention is paid to QC codes. For a QC code $\mathcal{C}$, we define its reversed code $\mathcal{R}$. We call $\mathcal{C}$ reversible or self-dual if $\mathcal{R}=\mathcal{C}$ or $\mathcal{C}^\perp=\mathcal{C}$, respectively. A formula for a GPM of $\mathcal{R}$ is given. We characterize GPMs for QC codes that combine reversibility and self-duality/self-orthogonality. For the reader interested in running computer search for optimal codes, we show the existence of binary self-orthogonal reversible QC codes that have the best known parameters as linear codes. These results can be obtained by brute-force search using GPMs that meet the above characterization.}

\section{Introduction}
In communication systems, such as mobile networks and data storage/retrieval systems, data is transmitted and stored as a sequence of symbols belonging to a particular set of alphabet. During data transmission or retrieval, noise may undoubtedly cause data errors and distortion. The goal of channel encoding/decoding is to remove noise effects as much as possible. Error correction techniques provide transmission robustness by detecting and correcting data distortion. Data is divided into groups of symbols. Each data group is encoded to a codeword by adding some redundancy symbols used by the decoder to detect and correct transmission errors, see \cite{Verlinde} for a detailed explanation of channel encoding/decoding. Formally, a code is defined as the set of all possible codewords. Different codes may have different lengths, sizes, or error correction capabilities. By defining a metric, e.g., the Hamming distance, the code with the largest minimum distance between its codewords has the greatest error-correcting capability. For fixed code length and code size, the code with the largest minimum Hamming distance is called an optimal code. Researchers use computer search to find optimal codes over different alphabets \cite{Eric,Grassl}. In \cite{Grassl}, one can find records of the best known minimum Hamming distances for different code parameters.

To facilitate software and hardware implementations, communication applications prefer the use of a code with logarithmic encoding/decoding techniques. Therefore, a preference has arisen for linear codes in different applications than non-linear codes. Moreover, linear codes with rich algebraic structure have found their superiority in engineering applications. For example, cyclic codes constitute a class of linear codes with a rich algebraic structure. This was followed by many generalizations for the algebraic structure of cyclic codes that led to other classes of linear codes with more complicated algebraic descriptions. Meanwhile, investigating codes over different alphabets has been of interest to many researchers. Although codes over finite fields are perhaps the most studied in literature, codes over various alphabets have also received considerable attention \cite{Honold1999,Spiegel1977,Cao2013,Dinh2004,QuangDinh2004}. For instance, constructing DNA codes in most cases requires four-element alphabets \cite{Siap2009,Liang2015}. The origin of the study of codes over finite commutative rings began with the study of codes over $\mathbb{Z}_4$, $\mathbb{F}_2[u]/\langle u^2\rangle$, and $\mathbb{F}_2[v]/\langle v^2+v\rangle$, see \cite{Hammons1994,FernndezCrdoba,Dougherty2001,Bonnecaze1999}. Since then, a sequence of generalizations of codes over different alphabets is followed. For instance, cyclic codes over the local principal ideal ring $\mathbb{Z}_{p^m}$ were investigated in \cite{Kanwar1997}, while codes over the principal ideal ring $\mathbb{Z}_m$ were considered in \cite{Spiegel1977}. Codes over general finite chain rings were examined in \cite{Norton2000}, where a chain ring is a local principal ideal ring. The Chinese remainder theorem is found as a powerful tool for decomposing a code over a finite commutative ring into a direct sum of codes over finite commutative local rings \cite{Dougherty2017}. However, in most cases, the commutative ring is chosen to be Frobenius since codes over Frobenius rings are shown in \cite{Wood1999,wood1997extension} to satisfy MacWilliams identity.

Although investigating properties of codes over different rings has been studied in depth, generalizing some classes of code that have rich algebraic structures to broader classes has been, and remains, of interest to coding theorists and mathematicians. For instance, cyclic codes have a noteworthy algebraic structure which has made this class subject to a sequence of generalizations. Some scientists consider code construction to be a purely mathematical discipline, however, decoding of codewords is still an important consideration in code design. The richer the algebraic structure of the code, the easier it is to obtain an encoding/decoding algorithm. Linear codes over a ring ${R}$ of length $n$ are described algebraically as ${R}$-submodules of ${R}^n$; hence, linear codes are vector spaces if  ${R}$ is a finite field. A cyclic code over ${R}$ is a linear code with the structure of an ideal in the ring of polynomials over ${R}$. The rich algebraic structure of cyclic codes over finite fields have encouraged researchers to generalize this class. The class of quasi-cyclic (QC) codes is obtained by not limiting the shift index of cyclic codes to unity. A QC code is a linear code invariant under the cyclic shift of a number of coordinates. QC codes have been addressed in several studies, e.g., \cite{Barbier2012,Conan1993,SanLing2001,Ling2003,Ling2005,Ling2006,Pei2007,Cao2010,Cayrel2010}. Another class of linear codes is the class of constacyclic codes, which is obtained by generalizing the shift constant of cyclic codes. In constacyclic codes, the shift constant is not restricted to unity; the shift constant of a constacyclic codes over a finite field can be any non-zero element. See \cite{Chen2012,Cao2013,Bakshi2012,LaGuardia2016} for a detailed description of constacyclic codes and their algebraic structure. Although the shift index of QC codes is not limited to unity, the block lengths are equal. Generalizing block lengths of QC codes, such that they are not necessarily equal, has led to a more general class of codes known in the literature as generalized quasi-cyclic (GQC) codes. In \cite{Gneri2017,Esmaeili2009,Cao2011,Siap05thestructure}, the algebraic structure of GQC codes is explained. However, generalizing the shift constant of QC codes, such that it does not necessarily equal unity, led to the class of quasi-twisted (QT) codes \cite{Jia2012,Gao2014}. In fact, QT codes generalize the shift constant of QC codes, and on the other hand they generalize the shift index of constacyclic codes. In \cite{Aydin2017}, QT and GQC codes are generalized to multi-twisted (MT) codes. MT codes are similar to GQC codes in that the block lengths are not necessarily equal, and they are similar to QT codes in that the shift constant is not necessarily equal to unity, moreover, the shift constants for different blocks are not necessarily the same. The algebraic structure of MT codes over finite fields is studied in \cite{Sharma2018,chauhan2021}.

We begin this chapter by defining cyclic and constacyclic codes over finite fields and presenting their algebraic structures as ideals in quotient rings. Cyclic and constacyclic codes have a unit shift index. Releasing the shift index of cyclic and constacyclic codes from being one leads to the classes of QC and QT codes, respectively. QC, QT, GQC, and MT codes over finite fields are invariant under some invertible linear transformations. Such invariance is the key behind defining an action on these codes which gives them the structure of modules over a principal ideal domain (PID) \cite{Lally2001,Conan1993,Matsui2015}. However, other algebraic structures are present in literature. In \cite{Cayrel2010}, QC codes over finite fields are viewed as cyclic codes over a non commutative ring of matrices over a finite field. In \cite{Cao_2011}, QC codes over finite fields are associated with some monic polynomial factors in the ring of polynomials over matrices with entries from the field. In \cite{Guneri2013}, the Chinese remainder theorem is used to provide the concatenated structure of QC codes, in which a QC code is written as a direct sum of concatenations of irreducible cyclic codes and linear codes. This concatenated structure has been generalized to QT and GQC codes in \cite{jia2012quasi,gao2015note,ozbudak2017structure}. We follow the representation of these codes as modules over PIDs. Specifically, we consider QC, QT, GQC, and MT codes as linear codes over PIDs. We show how generators of these codes can be deduced from their generator matrices. Hence, these generators are used to construct a generator polynomial matrix (GPM) for the code as a linear code over a PID. For the broader class of MT codes, we aim to provide a minimal generating set and an algorithmic technique for finding this minimal set. Since GPMs are over PIDs, we consider the reduced form of a matrix over a PID in the Section \ref{Matrices_PIDs}.

Modules over commutative rings generalize vector spaces over fields. When the modules are over PIDs, it was found that many known properties of vector spaces remain true for these modules. For instance, a submodule of a free module over a PID is free \cite{Roman2008-wr}. This makes MT codes free modules over PIDs, and we show how a basis of MT codes can be inferred. Although this basis is used to construct a GPM with entries from a PID, we aimed to find a unique reduced matrix form. This reduced form can be achieved by the Hermite normal form of matrices over PIDs \cite{gathen_gerhard_2013,Cohen2000lq,Kipp_Martin2012,Cohen1996}. The Hermite normal form of a matrix is a left equivalent matrix that generalizes the row reduced echelon form of matrices over a field. We show how to get the unique reduced GPM from any GPM. We prove the relationship between the code dimension and the diagonal entries of the reduced GPM. We present the identical equation satisfied by the GPM. The identical equation plays a major role in constructing the GPM of the dual code.

The dual of a code is defined as the set of all vectors that have a zero inner product with each codeword in the code. Duality, its properties, and its interrelationship with other branches of mathematics were among the research points that gained great importance in the past decades \cite{MacWilliams1978}. A code is self-dual if it is identical to its dual, while it is self-orthogonal if it is contained in its dual. Finding self-dual codes has been associated with many branches of mathematics, e.g., invariant theory \cite{Nebe2006}, combinatorics \cite{Pless1975}, design theory \cite{Assmus1992}, projective geometry \cite{Dougherty1234}, and lattices \cite{Conway1999}. We define dual codes, introduce the MacWilliams identity \cite{MacWilliams1978,Dougherty2017} that relates the weight enumerator of a code to the weight enumerator of its dual, and describe the dual code of any MT code which has been shown to be MT as well. Using the identical equation, we explain how to construct a GPM for the dual of a MT code. Specifying to subclasses of MT codes, formulas for GPMs of the dual codes of QC, QT, and GQC codes are presented. Intuitively, this provides conditions for the self-duality and self-orthogonality of these codes.

A code is called reversible if it is invariant under reversing the coordinates of its codewords. Reversibility is essential in some code design applications, for example, DNA codes \cite{Gaborit2005,Marathe2001,Milenkovic2006}, cryptography \cite{Li2017,Carlet2015}, and locally repairable codes \cite{Zeh2015}. Massey \cite{Massey1964} considered reversibility in the class of cyclic codes. He demonstrated that a cyclic code over a finite field is reversible if and only if its generator polynomial is self-reciprocal. Reversibility for codes over different alphabets has been considered in \cite{Bennenni2017,Guenda2013,Bayram2015}. Since QC codes generalize cyclic codes, it was natural to study reversibility in the class of QC codes. In \cite{TakiEldin2019,Eldin2020,TakiElDin2020}, reversibility for some classes of QC codes has been considered. We show an explicit formula for a GPM of the reversed code of a QC code, where the reversed code is the code obtained by reversing the coordinates of all codewords. Using this formula, we summarize some of the results of \cite{ELDIN2021} about the relations between reversibility, self-duality, and self-orthogonality of QC codes and their impact on the GPMs. We restrict ourselves to QC codes because this class contains many codes with the best known parameters \cite{SanLing2003,MartinezPerez2007,Greenough1992}, where a code has the best known parameters if its minimum distance meets the upper bound in \cite{Grassl}. Specifically, we aim to present some QC codes that combine reversibility, self-orthogonality, and optimality. We conclude this chapter by listing some results obtained by computer search for binary optimal self-orthogonal reversible QC codes in Table \ref{tab1}.

\section{From cyclic to multi-twisted codes}
\label{sec:1}
Let $\mathbb{F}_q$ be the finite field of order $q$, where $q$ is a prime power. If $q$ is a prime number, then $\mathbb{F}_q$ is isomorphic to the quotient ring $\mathbb{Z}/\langle q\rangle$, where $\langle q\rangle=\{0,\pm q,\pm 2q,\pm 3q, \ldots\}$ is the ideal generated by $q$ in the ring of integers $\mathbb{Z}$. In $\mathbb{Z}/\langle q\rangle$, additions and multiplications are carried modulo $q$. However, if $q=p^d$ for a prime $p$ and an integer $d>1$, then $\mathbb{F}_q$ is isomorphic to the quotient ring $\mathbb{F}_p[x]/\langle f(x)\rangle$, where $\mathbb{F}_p[x]$ is the ring of polynomials in the indeterminate $x$ with coefficient from $\mathbb{F}_p$, $f(x) \in \mathbb{F}_p[x]$ is an irreducible polynomial of degree $d$, and $\langle f(x)\rangle=\{a(x)f(x)|a(x)\in\mathbb{F}_p[x]\}$.

A code $\mathcal{C}$ over the alphabet $\mathbb{F}_q$ of length $n$ is a subset of the vector space $\mathbb{F}_q^n$. The elements of $\mathcal{C}$ are called codewords. The weight of a codeword of $\mathcal{C}$ is the number of non-zero coordinates in the codeword. The minimum weight of $\mathcal{C}$ is the smallest weight among all non-zero codewords. The Hamming distance between two codewords of $\mathcal{C}$ is the number of coordinates at which the two codewords differ. The minimum Hamming distance of $\mathcal{C}$, denoted $d_{\mathrm{min}}(\mathcal{C})$ or $d_{\mathrm{min}}$, is the smallest Hamming distance between each pair of unequal codewords. If $\mathcal{C}$ is a code over $\mathbb{F}_q$ of length $n$ and minimum distance $d_{\mathrm{min}}$, then the first $n-d_{\mathrm{min}}+1$ coordinates of all the codewords are distinct. Hence, there are at most $q^{n-d_{\mathrm{min}}+1}$ codewords in $\mathcal{C}$. This gives the following fundamental bound on the minimum distance of a code, which is known as the Singleton bound \cite{MacWilliams1978}.
\begin{Thm}[Singleton Bound]
Let $\mathcal{C}$ be a code of length $n$ over an alphabet of size $q$. If $\mathcal{C}$ has a minimum distance $d_{\mathrm{min}}$, then
\begin{equation*}
|\mathcal{C}|\le q^{n-d_{\mathrm{min}}+1}.
\end{equation*}
\end{Thm}

A code $\mathcal{C}$ is linear if it is a subspace of $\mathbb{F}_q^n$, that is, $\mathcal{C}$ is closed under addition and scalar multiplication by elements of $\mathbb{F}_q$. The dimension $k$ of a linear code is its dimension as a vector space over $\mathbb{F}_q$. Therefore, the size of a linear code of dimension $k$ is $|\mathcal{C}|=q^k$. The minimum weight and the minimum Hamming distance are equal for any linear code, and the Singleton bound takes the form
\begin{equation}
\label{singleton}
d_{\mathrm{min}}\le n-k+1.
\end{equation} 
There are other minimum distance bounds, for example, sphere-packing bound \cite{Dougherty2017} and Gilbert–Varshamov bound \cite{Gilbert1952}. For fixed alphabet, code length, and dimension, a code with the largest minimum Hamming distance is best for practical communication systems. A code that achieves the upper bound of any of the bounds on the minimum distance is called optimal. In particular, a code that achieves the Singleton bound is called maximum distance separable (MDS), while a code that achieves the sphere-packing bound is called perfect. Researchers are interested in using computer searches to find optimal codes. Their results yield codes with the best known $d_{\mathrm{min}}$ for fixed $n$ and $k$. These records can be found in \cite{Grassl}.

A linear code $\mathcal{C}$ over $\mathbb{F}_q$ of length $n$ is cyclic if it is invariant under cyclic shifts of its coordinates. Namely, 
\begin{equation*}
\left(c_0, c_1, \ldots, c_{n-2}, c_{n-1} \right)\in \mathcal{C} \Rightarrow \left(c_{n-1}, c_0, \ldots, c_{n-3}, c_{n-2} \right)\in \mathcal{C}.
\end{equation*}
The main advantage of using cyclic codes over other codes in communication systems is that cyclic codes can be efficiently encoded using shift registers. Let $\mathscr{R}$ be the quotient ring $\mathbb{F}_q[x] / \langle x^n-1 \rangle$. Elements of $\mathscr{R}$ are represented by polynomials in the indeterminate $x$ over $\mathbb{F}_q$ of degree at most $n-1$ with addition and multiplication carried out modulo $x^n-1$. In addition, $\mathscr{R}$ can be viewed as a vector space over $\mathbb{F}_q$ of dimension $n$, so $\mathscr{R}$ is isomorphic to $\mathbb{F}_q^n$. It is usual to use this isomorphism to represent codewords by polynomials rather than vectors. Precisely, the word $\left(a_0, a_1, \ldots, a_{n-2}, a_{n-1} \right)\in\mathbb{F}_q^n$ has a  polynomial representation $a_0 + a_1 x +\cdots + a_{n-2} x^{n-2} + a_{n-1} x^{n-1}$ for its correspondence in $\mathscr{R}$. In its polynomial representation, a linear code is an $\mathbb{F}_q$-subspace of $\mathscr{R}$. Since multiplication is carried modulo $x^n-1$, cyclic shift of a codeword corresponds to multiplying its polynomial representation by $x$. The ring $\mathscr{R}$ can be thought of as an $\mathscr{R}$-module. Therefore, the property of a cyclic code being invariant under cyclic shifts makes this code acts as an $\mathscr{R}$-submodule of $\mathscr{R}$, i.e., an ideal in $\mathscr{R}$. There is a one-to-one correspondence between cyclic codes over $\mathbb{F}_q$ of length $n$ and ideals of $\mathscr{R}$. This property is generalized to a broader class, the class of constacyclic codes. For a non-zero $\lambda \in \mathbb{F}_q$, a linear code $\mathcal{C}$ is $\lambda$-constacyclic if 
\begin{equation*}
\left(c_0, c_1, \ldots, c_{n-2}, c_{n-1} \right)\in \mathcal{C} \Rightarrow \left(\lambda c_{n-1}, c_0, \ldots, c_{n-3}, c_{n-2} \right)\in \mathcal{C}.
\end{equation*}
We call $\lambda$ the shift constant. A $1$-constacyclic code is simply a cyclic code. A $(-1)$-constacyclic code is called negacyclic; any negacyclic code over a ring of characteristic $2$ is cyclic. The class of constacyclic codes includes cyclic and negacyclic codes. Analogous to cyclic codes, constacyclic codes have a polynomial representation. Consequently, a $\lambda$-constacyclic code over $\mathbb{F}_q$ of length $n$ is an ideal in the quotient ring $\mathscr{R}_\lambda=\mathbb{F}_q[x] / \langle x^n-\lambda \rangle$, and there is a one-to-one correspondence between $\lambda$-constacyclic codes and ideals of $\mathscr{R}_\lambda$.

A PID is an integral domain in which each ideal is principal. That is, an integral domain ${R}$ is a PID if for every ideal $I$ of ${R}$, there is an element $a \in I$ such that $I={R} a=\{ra|r\in R\}$. The element $a$ is a generator of $I$ and we write $I=\langle a\rangle$. Examples of PIDs are the ring of integers $\mathbb{Z}$, the quotient ring $\mathbb{Z}[x]/\langle x^2+1\rangle$, and the ring of polynomials $\mathbb{F}_q[x]$. PIDs play an important role in the algebraic structure of some classes of codes. For instance, for constacyclic codes, we have the following:
\begin{Prop}
Let $\mathcal{C}$ be a $\lambda$-constacyclic code over $\mathbb{F}_q$ of length $n$ and dimension $k$. Then $\mathcal{C}=\langle g(x) \rangle$ for some $g(x)\in \mathbb{F}_q[x]$ such that $g(x)|\left(x^n-\lambda\right)$. Moreover, $k=n-\deg(g(x))$.\end{Prop}
\begin{proof}
We know that $\mathcal{C}$ is an ideal in $\mathscr{R}_\lambda$. However, the projection map $\pi: \mathbb{F}_q[x] \to \mathscr{R}_\lambda$ defines a bijection between ideals of $\mathbb{F}_q[x]$ that contain $\langle x^n-\lambda \rangle$ and ideals of $\mathscr{R}_\lambda$. Then, $\mathcal{C}$ can be shown as an ideal of $\mathbb{F}_q[x]$ containing $\langle x^n-\lambda \rangle$. Since $\mathbb{F}_q[x]$ is a PID, $\mathcal{C} =\langle g(x) \rangle \supseteq \langle x^n-\lambda \rangle$ for some $g(x)\in \mathbb{F}_q[x]$. But $\langle g(x) \rangle \supseteq \langle x^n-\lambda \rangle$ if and only if $g(x)|\left( x^n-\lambda\right)$, i.e., $a(x)g(x)= x^n-\lambda$ for some $a(x) \in \mathbb{F}_q[x]$. We call $a(x)$ the check polynomial of $\mathcal{C}$. The dimension of $\mathcal{C}$ is the dimension of $\langle g(x)\rangle / \langle x^n-\lambda \rangle$ as an $\mathbb{F}_q$-vector space. Therefore, $k=\deg(a(x))=n-\deg(g(x))$. 
\end{proof}

The polynomial $g(x)$ is not uniquely specified; any $g(x)$ such that $\mathcal{C}=\langle g(x) \rangle$ is called a generator polynomial of $\mathcal{C}$. However, restricting $g(x)$ to be monic ensures its uniqueness. Thus, one can uniquely define $g(x)$ to be the least degree monic polynomial in $\mathcal{C}$ as an ideal in $\mathbb{F}_q[x]$. So far we have shown that a $\lambda$-constacyclic code $\mathcal{C}$ is completely determined by the monic polynomial $g(x)\in\mathcal{C}$ that divides $x^n-\lambda$ and generates $\mathcal{C}$ as an ideal in the PID $\mathbb{F}_q[x]$. Consequently, all $\lambda$-constacyclic codes over $\mathbb{F}_q$ of length $n$ can be enumerated by decomposing $ x^n-\lambda $ into irreducible factors in $\mathbb{F}_q[x]$. We aim to give the corresponding unique representation in the broader class of multi-twisted codes. 

\begin{exmp}
Let $\mathbb{F}_9=\{0,1,2,\theta,2\theta,1+\theta,1+2\theta,2+\theta,2+2\theta\}$, where $\theta^2+2\theta+2=0$. We construct a $2$-constacyclic code $\mathcal{C}$ over $\mathbb{F}_9$ of length $5$. In $\mathbb{F}_9$, $2=-1$, then $\mathcal{C}$ is negacyclic as well. Factorizing $x^5-2$ to irreducible factors yields $x^5-2=\left(x+1 \right)\left(x^2+2\theta x+1 \right)\left(x^2+(\theta+2)x+1 \right)$. Different $2$-constacyclic codes over $\mathbb{F}_9$ of length $5$ are obtained through different choices of generator polynomials from the factorization of $x^5-2$ in $\mathbb{F}_9[x]$. Let $\mathcal{C}=\langle g(x) \rangle$, where $g(x)= x^2+2\theta x+1$. The dimension of $\mathcal{C}$ is $k=n-\deg\left(g(x)\right)=3$. The check polynomial of $\mathcal{C}$ is $a(x)=\frac{x^5-2}{g(x)}=\left(x+1 \right)\left(x^2+(\theta+2)x+1 \right)$. By listing all codewords of $\mathcal{C}$, one can observe that the codewords with the fewest number of monomials in their polynomial representation contain three monomials. Therefore, $d_{\mathrm{min}}=3$. From \eqref{singleton}, $\mathcal{C}$ achieves the Singleton bound, and is therefore MDS.\end{exmp}

Constacyclic codes generalize cyclic codes by generalizing the shift constant $\lambda$. A different generalization of cyclic codes is obtained by generalizing the shift index, and this leads to the class of QC codes.
\begin{Def}
A linear code $\mathcal{C}$ over $\mathbb{F}_q$ of length $n$ is $\ell$-QC if it is invariant under cyclic shift by $\ell$ coordinates. Namely,
\begin{equation}
\label{shift1}
\left( c_{1}, c_{2}, \ldots, c_{n}\right) \in\mathcal{C} \Rightarrow \left( c_{n-\ell+1}, \dots, c_{n}, c_{1}, c_{2},\ldots, c_{n-\ell}\right) \in\mathcal{C}.
\end{equation}\end{Def}

Define the index of a QC code to be the smallest positive integer $\ell$ satisfying \eqref{shift1}. Hereinafter, the index is denoted by $\ell$. 
\begin{Prop}
\label{Prop_index}
The index $\ell$ of a QC code divides the code length $n$.\end{Prop}
\begin{proof}
Let $m$ be the unique positive integer defined such that $0\le m\ell -n  < \ell$. Since $\mathcal{C}$ is invariant under $\ell$ shifts of coordinates, it is invariant under $m\ell -n$ shifts. Then $m\ell -n  =0$ because $m\ell -n  < \ell$ and $\ell$ is the smallest integer satisfying \eqref{shift1}. Thus $n=m\ell$ and $\ell|n$. 
\end{proof}

The integer $m=n/\ell$ in the proof of Proposition \ref{Prop_index} is called the co-index of $\mathcal{C}$. A codeword of an $\ell$-QC code $\mathcal{C}$ of length $n$ and co-index $m$ can be partitioned as follows:
\begin{equation}
\label{shift2}
\mathbf{c}=\left( c_{0,1}, c_{0,2}, \ldots, c_{0,\ell}, c_{1,1}, c_{1,2}, \ldots, c_{1,\ell}, \ldots, c_{m-1,1}, c_{m-1,2}, \ldots, c_{m-1,\ell} \right).
\end{equation}
A linear code $\mathcal{C}$ is QC if and only if, for every $\mathbf{c}\in\mathcal{C}$ in the form of \eqref{shift2}, $\mathcal{C}$ contains the codeword
\begin{equation*}
\left( c_{m-1,1}, c_{m-1,2}, \ldots, c_{m-1,\ell}, c_{0,1}, c_{0,2}, \ldots, c_{0,\ell}, \ldots, c_{m-2,1}, c_{m-2,2}, \ldots, c_{m-2,\ell} \right).
\end{equation*}
Let $T_\ell$ be the automorphism of $\mathbb{F}_q^n$ that corresponds to shifting by $\ell$ coordinates. That is, $T_\ell:\left( a_{1}, a_{2}, \ldots, a_{n}\right) \mapsto \left( a_{n-\ell+1}, \dots, a_{n}, a_{1}, a_{2},\ldots, a_{n-\ell}\right)$ for every $\left( a_{1}, a_{2}, \ldots, a_{n}\right)\in\mathbb{F}_q^n$. In the standard basis of $\mathbb{F}_q^n$, $T_\ell$ has the $n\times n$ matrix representation
\begin{equation*}
\begin{pmatrix}
\mathbf{0}_{\ell\times(n-\ell)} & \mathbf{I}_{\ell\times\ell}\\
\mathbf{I}_{(n-\ell)\times(n-\ell)} & \mathbf{0}_{(n-\ell)\times\ell}
\end{pmatrix},
\end{equation*} 
where $\mathbf{0}_{i\times j}$ and $\mathbf{I}_{i}$ are the zero matrix of size $i\times j$ and the identity matrix of size $i\times i$, respectively. An alternative definition of an $\ell$-QC code is a linear subspace of $\mathbb{F}_q^n$ that is invariant under $T_\ell$. This invariance is used to give QC codes the structure of module over PID. Although QC codes generalize cyclic codes because cyclic codes are QC with $\ell=1$, QC codes do not generalize constacyclic codes. One can find a class containing QC and constacyclic codes as subclasses by generalizing the linear operator $T_\ell$. For $0\ne \lambda\in\mathbb{F}_q$, let $T_{(\ell,\lambda)}$ be the automorphism of $\mathbb{F}_q^n$ such that $T_{(\ell,\lambda)}:\left( a_{1}, a_{2}, \ldots, a_{n}\right) \mapsto \left(\lambda a_{n-\ell+1}, \dots, \lambda a_{n}, a_{1}, a_{2},\ldots, a_{n-\ell}\right)$. In the standard basis of $\mathbb{F}_q^n$, $T_{(\ell,\lambda)}$ has the $n\times n$ matrix representation 
\begin{equation}
\label{M_ell_lambda}
\mathbf{M}=\begin{pmatrix}
\mathbf{0}_{\ell\times(n-\ell)} & \lambda\mathbf{I}_{\ell}\\
\mathbf{I}_{(n-\ell)} & \mathbf{0}_{(n-\ell)\times\ell}
\end{pmatrix}.
\end{equation} 
Obviously $T_\ell=T_{(\ell,1)}$, and this generalizes QC codes to QT codes. From \eqref{M_ell_lambda}, we have $\mathbf{M}^m=\lambda \mathbf{I}_{n}$, hence $\frac{1}{\lambda}T_{(\ell,\lambda)}^m$ is the identity map on $\mathbb{F}_q^n$. Then $\mathbf{M}^{mt}= \mathbf{I}_{n}$ and $T_{(\ell,\lambda)}^{mt}$ is the identity map, where $t$ is the multiplicative order of $\lambda$, i.e., $t$ is the least positive integer such that $\lambda^t=1$.

\begin{Def}
For a non-zero $\lambda \in\mathbb{F}_q$, a linear code $\mathcal{C}$ is called $(\ell,\lambda)$-QT if $\mathcal{C}$ is invariant under $T_{(\ell,\lambda)}$. That is,
\begin{equation}
\label{shift3}
\left( c_{1}, c_{2}, \ldots, c_{n}\right) \in\mathcal{C} \Rightarrow \left( \lambda c_{n-\ell+1}, \dots, \lambda c_{n}, c_{1}, c_{2},\ldots, c_{n-\ell}\right) \in\mathcal{C}.
\end{equation}
The index of $\mathcal{C}$ is the smallest positive integer $\ell$ that satisfies \eqref{shift3}. We call $\lambda$ the shift constant of $\mathcal{C}$.\end{Def}

Similar to QC codes, the index of a QT code divides its length. The co-index of $\mathcal{C}$ is the integer $m=n/\ell$. By partitioning the codewords of a linear code $\mathcal{C}$ of length $m\ell$ to $m$ blocks of length $\ell$ each, then $\mathcal{C}$ is $(\ell,\lambda)$-QT if and only if 
\begin{equation*}
\left( \lambda c_{m-1,1},\!  \lambda c_{m-1,2}, \ldots, \! \lambda c_{m-1,\ell}, c_{0,1}, c_{0,2}, \ldots, c_{0,\ell}, \ldots, c_{m-2,1}, c_{m-2,2}, \ldots, c_{m-2,\ell} \right)
\end{equation*}
is a codeword in $\mathcal{C}$ for every codeword in the form of \eqref{shift2}. The class of QT codes generalizes QC and constacyclic codes because an $\ell$-QC code is $(\ell,1)$-QT, a $\lambda$-constacyclic code is $(1,\lambda)$-QT, and a cyclic code is $(1,1)$-QT.

\begin{Thm}
\label{Cond_QT}
Let $\mathcal{C}$ be a linear code over $\mathbb{F}_q$ of length $n$, dimension $k$, and a $k\times n$ generator matrix $\mathbf{G}$. Then, $\mathcal{C}$ is $(\ell,\lambda)$-QT if and only if the rank of the block matrix 
\begin{equation}
\label{Condition_QT}
\begin{pmatrix}
\mathbf{G}\\
\mathbf{G} \mathbf{M}^t
\end{pmatrix}
\end{equation} 
is $k$, where $^t$ stands for matrix transpose and $\mathbf{M}$ is given by \eqref{M_ell_lambda}.\end{Thm}
\begin{proof}
{Assume $\mathbf{G}$ satisfies  
\begin{equation*} \text{rank}\begin{pmatrix}
\mathbf{G}\\
\mathbf{G} \mathbf{M}^t
\end{pmatrix} =k.
\end{equation*} 
For any codeword $\mathbf{c}\in\mathcal{C}$, there exists $\mathbf{a}\in\mathbb{F}_q^k$ such that $\mathbf{c}=\mathbf{a}\mathbf{G}$. Since the dimension of $\mathcal{C}$ is $k$, $\text{rank}(\mathbf{G})=k$ and 
\begin{equation*}
\text{rank}\begin{pmatrix}
\mathbf{G}\\
\mathbf{c} \mathbf{M}^t
\end{pmatrix} \ge k.
\end{equation*} 
But
\begin{equation*}
\text{rank}\begin{pmatrix}
\mathbf{G}\\
\mathbf{c} \mathbf{M}^t
\end{pmatrix}=\text{rank}\begin{pmatrix}
\mathbf{G}\\
\mathbf{a}\mathbf{G} \mathbf{M}^t
\end{pmatrix} = \text{rank}\left(\begin{pmatrix}
\mathbf{I}_{k} & \mathbf{0}_{k \times k}\\
\mathbf{0}_{1\times k} & \mathbf{a}
\end{pmatrix}\begin{pmatrix}
\mathbf{G}\\
\mathbf{G} \mathbf{M}^t
\end{pmatrix}\right) \le \text{rank}\begin{pmatrix}
\mathbf{G}\\
\mathbf{G} \mathbf{M}^t
\end{pmatrix} =k.
\end{equation*} 
Then
\begin{equation*}
\text{rank}\begin{pmatrix}
\mathbf{G}\\
\mathbf{c} \mathbf{M}^t
\end{pmatrix} = k.
\end{equation*}
This shows that $\mathbf{c}\mathbf{M}^t \in\mathcal{C}$ and $\mathcal{C}$ is $(\ell,\lambda)$-QT because it is invariant under $T_{\ell,\lambda}$. 

Conversely, if $\mathcal{C}$ is $(\ell,\lambda)$-QT, then rows of $\mathbf{G} \mathbf{M}^t$ are codewords and 
\begin{equation*} \text{rank}\begin{pmatrix}
\mathbf{G}\\
\mathbf{G} \mathbf{M}^t
\end{pmatrix} =k.
\end{equation*} }
\end{proof}

\begin{exmp}
\label{FirstQTCode}
Let $\mathcal{C}$ be the linear code over $\mathbb{F}_4$ of length $n=9$, dimension $k=6$, and systematic generator matrix
\begin{equation}
\label{GMinExample}
\mathbf{G}=\begin{pmatrix}
   1&0&0&0&0&0&0&1+\theta&1\\
   0&1&0&0&0&0&1+\theta&\theta&1\\
   0&0&1&0&0&0&1&0&\theta\\
   0&0&0&1&0&0&1&\theta&\theta\\
   0&0&0&0&1&0&0&\theta&1\\
   0&0&0&0&0&1&1+\theta&1&1+\theta
\end{pmatrix},
\end{equation}
where $\theta^2+\theta+1=0$. By Theorem \ref{Cond_QT}, $\mathcal{C}$ is $(3,\theta)$-QT code because 
\begin{equation*}
\label{Condition_QT}
\text{rank}\begin{pmatrix}
\mathbf{G}\\
\mathbf{G} \mathbf{M}^t
\end{pmatrix}=\text{rank}\begin{pmatrix}
   1&0&0&0&0&0&0&1+\theta&1\\
   0&1&0&0&0&0&1+\theta&\theta&1\\
   0&0&1&0&0&0&1&0&\theta\\
   0&0&0&1&0&0&1&\theta&\theta\\
   0&0&0&0&1&0&0&\theta&1\\
   0&0&0&0&0&1&1+\theta&1&1+\theta\\
   0&1&\theta&1&0&0&0&0&0\\
   1&1+\theta&\theta&0&1&0&0&0&0\\
   \theta&0&1+\theta&0&0&1&0&0&0\\
   \theta&1+\theta&1+\theta&0&0&0&1&0&0\\
   0&1+\theta&\theta&0&0&0&0&1&0\\
   1&\theta&1&0&0&0&0&0&1
\end{pmatrix}=6.\end{equation*} 
In fact, $\mathcal{C}$ is optimal because its minimum distance is $d_{\mathrm{min}}=3$, which is the best known minimum distance for a linear code over $\mathbb{F}_4$ of length $9$ and dimension $6$, see \cite{Grassl}.\end{exmp}

The linear transformation $T_{\ell,\lambda}$ is an automorphism of the vector space $\mathbb{F}_q^{m\ell}$. Define an $\mathbb{F}_q[x]$-module structure on $\mathbb{F}_q^{m\ell}$ such that the action of $x$ on $\mathbb{F}_q^{m\ell}$ is the action of $T_{\ell,\lambda}$. Specifically, the action of $f(x)=\sum_{h}f_h x^h\in\mathbb{F}_q[x]$ on $\mathbf{a}\in\mathbb{F}_q^{m\ell}$ is defined by $f(x) \mathbf{a}=\sum_h f_h T_{\ell,\lambda}^h \left( \mathbf{a}\right)$. Let $\mathcal{C}$ be an $(\ell,\lambda)$-QT code over $\mathbb{F}_q$ of length $m\ell$. Since $\mathcal{C}$ is a $T_{\ell,\lambda}$-invariant $\mathbb{F}_q$-subspace of $\mathbb{F}_q^{m\ell}$, it is an $\mathbb{F}_q[x]$-submodule of $\mathbb{F}_q^{m\ell}$. The opposite is also true; meaning that an $\mathbb{F}_q[x]$-submodule of $\mathbb{F}_q^{m\ell}$ is an $(\ell,\lambda)$-QT code over $\mathbb{F}_q$ of length $m\ell$. In addition, the following result shows that an $(\ell,\lambda)$-QT code corresponds to an $\mathbb{F}_q[x]$-submodule of $\mathscr{R}_\lambda^\ell$, where $\mathscr{R}_\lambda=\mathbb{F}_q[x]/\langle x^m-\lambda \rangle$. This correspondence leads to the definition of the polynomial representation of $(\ell,\lambda)$-QT codes.

\begin{Thm}
\label{QT_Corresp}
There is a one-to-one correspondence between $(\ell,\lambda)$-QT codes over $\mathbb{F}_q$ of length $m\ell$ and $\mathbb{F}_q[x]$-submodules of $\mathscr{R}_\lambda^\ell$, where $\mathscr{R}_\lambda=\mathbb{F}_q[x]/\langle x^m-\lambda \rangle$.\end{Thm}
\begin{proof}
Let $\phi: \mathbb{F}_q^{m\ell} \rightarrow \mathscr{R}_\lambda^\ell$ be the $\mathbb{F}_q$-vector space isomorphism such that
\begin{eqnarray*}
\phi:\mathbf{a}=\left( a_{0,1}, a_{0,2}, \ldots, a_{0,\ell}, a_{1,1}, a_{1,2}, \ldots, a_{1,\ell}, \ldots, a_{m-1,1}, a_{m-1,2}, \ldots, a_{m-1,\ell} \right)\\
\mapsto \mathbf{a}(x)=\left( a_1(x), a_2(x), \ldots, a_\ell(x)\right),
\end{eqnarray*}
where $a_j(x)=a_{0,j}+a_{1,j} x+a_{2,j} x^2+\cdots +a_{m-1,j}x^{m-1}\in \mathscr{R}_\lambda$ for $1\le j \le \ell$. 
Consider the following diagram of $\mathbb{F}_q$-vector space isomorphisms
\begin{equation}
\label{Commut_diagram}
\begin{tikzcd}[swap]
    \mathbb{F}_q^{m\ell} \arrow{r}[swap]{\phi}{} \arrow{d}{T_{\ell,\lambda}}  & \mathscr{R}_\lambda^\ell \arrow{d}[swap]{\psi}{} \\  
    \mathbb{F}_q^{m\ell} \arrow{r}[swap]{}{\phi}   & \mathscr{R}_\lambda^\ell
\end{tikzcd}
\end{equation}
where $\psi:  \left( a_1(x), a_2(x), \ldots, a_\ell(x)\right) \mapsto \left(x a_1(x),x a_2(x), \ldots,x a_\ell(x)\right)$. Diagram \eqref{Commut_diagram} is commutative because
\begin{align*}
\phi\!\circ\! T_{\ell,\lambda}(\mathbf{a})\! &=\phi\left( \lambda a_{m-1,1}, \lambda a_{m-1,2}, \ldots, \lambda a_{m-1,\ell}, a_{0,1}, a_{0,2}, \ldots, a_{0,\ell},\ldots,\right. \\ &\qquad \qquad \qquad \qquad \qquad \qquad \qquad \qquad\qquad  \left. a_{m-2,1}, a_{m-2,2}, \ldots, a_{m-2,\ell} \right) \\
&=\left( \lambda a_{m-1,1}+a_{0,1} x+a_{1,1} x^2+\cdots +a_{m-2,1}x^{m-1},\ldots, \right. \\ & \ \ \ \quad \quad \qquad \qquad \qquad\qquad   \left. \lambda a_{m-1,\ell}+a_{0,\ell} x+a_{1,\ell} x^2+\cdots +a_{m-2,\ell}x^{m-1}\right) \\
&=\left(  a_{m-1,1}x^m+a_{0,1} x+a_{1,1} x^2+\cdots +a_{m-2,1}x^{m-1}, \ldots, \right. \\ &\qquad \qquad \qquad \qquad\qquad   \left. a_{m-1,\ell}x^m+a_{0,\ell} x+a_{1,\ell} x^2+\cdots +a_{m-2,\ell}x^{m-1}\right) \\
&=x\!\left(\!  a_{0,1} +a_{1,1} x+\cdots +a_{m-1,1}x^{m-1}\!, \ldots,  a_{0,\ell}+a_{1,\ell} x+\cdots +a_{m-1,\ell}x^{m-1}\!\right) \\
&=x\phi\left(\mathbf{a}\right)=\psi\circ\phi\left(\mathbf{a}\right).
\end{align*}
For any $f(x)=\sum_{h}f_h x^h\in\mathbb{F}_q[x]$, we have 
\begin{equation*}
f(x)\phi\left(\mathbf{a}\right)=\sum_{h}f_h x^h \phi\left(\mathbf{a}\right)= \sum_{h}f_h \psi^h\circ \phi\left(\mathbf{a}\right)= \sum_{h}f_h \phi\circ T_{\ell,\lambda}^h\left(\mathbf{a}\right)=\phi\left( f(T_{\ell,\lambda})(\mathbf{a})\right).
\end{equation*}
Then $\phi$ is an $\mathbb{F}_q[x]$-module isomorphism, where the module structure of $\mathbb{F}_q^{m\ell}$ is described in the paragraph before Theorem \ref{QT_Corresp}. This isomorphism determines the correspondence between submodules of $\mathscr{R}_\lambda^\ell$ and submodules of $\mathbb{F}_q^{m\ell}$, i.e., $(\ell,\lambda)$-QT codes over $\mathbb{F}_q$ of length $m\ell$.
\end{proof}

Let $\mathcal{C}$ be an $(\ell,\lambda)$-QT code over $\mathbb{F}_q$ of length $m\ell$. Then $\mathcal{C}$ is an $\mathbb{F}_q[x]$-submodule of $\mathbb{F}_q^{m\ell}$, and $f(x)\in\mathbb{F}_q[x]$ acts on a codeword $\mathbf{c}\in\mathcal{C}$ by $f(T_{\ell,\lambda}) \left(\mathbf{c}\right)$. Since QT codes generalize the constacyclic codes, we aim to provide a polynomial representation for $(\ell,\lambda)$-QT codes similar to that of cyclic and constacyclic codes. The polynomial representation of an $(\ell,\lambda)$-QT code $\mathcal{C}\subseteq \mathbb{F}_q^{m\ell}$ is $\phi\left(\mathcal{C}\right)$, its image under the isomorphism defined in the proof of Theorem \ref{QT_Corresp}. Specifically, the codeword $\mathbf{c}$ given by \eqref{shift2} is represented by the polynomial vector
\begin{equation}
\label{Poly_rep}
\begin{split}
\mathbf{c}(x)=&\left( c_{0,1}+c_{1,1} x+c_{2,1} x^2+\cdots +c_{m-1,1}x^{m-1}, \right.\\
&\qquad\qquad\qquad c_{0,2}+c_{1,2} x+c_{2,2} x^2+\cdots +c_{m-1,2}x^{m-1}, \ldots,\\ 
&\qquad\qquad\qquad\qquad\qquad\quad \left. c_{0,\ell}+c_{1,\ell} x+c_{2,\ell} x^2+\cdots +c_{m-1,\ell}x^{m-1}\right)\in \mathscr{R}_\lambda^\ell.
\end{split}
\end{equation}
We do not distinguish between the code and its polynomial representation. That is, an $(\ell,\lambda)$-QT code $\mathcal{C}$ over $\mathbb{F}_q$ of length $m\ell$ is an $\mathbb{F}_q[x]$-submodule of $\mathscr{R}_\lambda^\ell$. Hence, $\mathcal{C}$ is a torsion $\mathbb{F}_q[x]$-module \cite{Roman2008-wr}. In fact, $\left(x^m-\lambda\right)$ annihilates all codewords $\mathbf{c}(x)\in\mathcal{C}$, i.e., $\left(x^m-\lambda\right) \mathbf{c}(x)=\mathbf{0}$. Then, the ideal $\langle x^m-\lambda \rangle \subset \mathbb{F}_q[x]$ is in the annihilator of $\mathcal{C}$. This allows $\mathcal{C}$ to be seen as an $\mathscr{R}_\lambda$-module. However, our interest will be in the $\mathbb{F}_q[x]$-module structure of $\mathcal{C}$. But more than that, we need to make $\mathcal{C}$ into an $\mathbb{F}_q[x]$-submodule of $\left(\mathbb{F}_q[x]\right)^\ell$; the following result leads to this.

\begin{Thm}
\label{QT_Corresp2}
There is a one-to-one correspondence between $(\ell,\lambda)$-QT codes over $\mathbb{F}_q$ of length $m\ell$ and $\mathbb{F}_q[x]$-submodules of $\left(\mathbb{F}_q[x]\right)^\ell$ that contain the submodule
\begin{equation*}
M=\left( (x^m-\lambda)\mathbb{F}_q[x] \right)^\ell.
\end{equation*}
\end{Thm}
\begin{proof}
$\left(\mathbb{F}_q[x]\right)^\ell$ is a free $\mathbb{F}_q[x]$-module of rank $\ell$ and $M$ is the submodule
\begin{eqnarray*}
M=\left\{\big((x^m-\lambda)f_1(x), (x^m-\lambda)f_2(x), \ldots, (x^m-\lambda)f_\ell(x)\big) \qquad \qquad \qquad \qquad\quad \right. \\
\left. \text{such that } f_j(x)\in\mathbb{F}_q[x] \text{ for } 1\le j\le \ell\right\}.
\end{eqnarray*}
Let $\pi:\left(\mathbb{F}_q[x]\right)^\ell \rightarrow \left(\mathbb{F}_q[x]\right)^\ell \! \slash M$ be the projection homomorphism. Then, $\pi$ defines a one-to-one correspondence between submodules of $\left(\mathbb{F}_q[x]\right)^\ell$ that contain $M$ and submodules of $\left(\mathbb{F}_q[x]\right)^\ell \! \slash M$. In addition, let $\tau$ be the natural $\mathbb{F}_q[x]$-module isomorphism between $\left(\mathbb{F}_q[x]\right)^\ell \! \slash M$ and $\mathscr{R}_\lambda^\ell$. Then, $\tau$ induces one-to-one correspondence between submodules of $\mathscr{R}_\lambda^\ell$ and the submodules of $\left(\mathbb{F}_q[x]\right)^\ell$ that contain $M$. The result follows from Theorem \ref{QT_Corresp}.
\end{proof}

Theorem \ref{QT_Corresp} presents any $(\ell,\lambda)$-QT code $\mathcal{C}$ as an $\mathbb{F}_q[x]$-submodule of $\mathscr{R}_\lambda^\ell$. While Theorem \ref{QT_Corresp2} presents $\mathcal{C}$ as an $\mathbb{F}_q[x]$-submodule of $\left(\mathbb{F}_q[x]\right)^\ell$. The proof of Theorem \ref{QT_Corresp2} provides a way to get generators of $\mathcal{C}$ as a submodule of $\left(\mathbb{F}_q[x]\right)^\ell$ from those generating $\mathcal{C}$ as a submodule of $\mathscr{R}_\lambda^\ell$. Specifically, let $\mathcal{C}$ be generated as a submodule of $\mathscr{R}_\lambda^\ell$ by the set
\begin{equation*}
\left\{ \left( g_{i,1}(x)+\langle x^m-\lambda \rangle , g_{i,2}(x)+\langle x^m-\lambda \rangle ,\ldots, g_{i,\ell}(x)+\langle x^m-\lambda \rangle \right)\  \big| \ 1\le i \le r \right\},
\end{equation*}
where $g_{i,j}(x)\in\mathbb{F}_q[x]$ for $1\le i\le r$ and $1\le j\le \ell$. Then $\mathcal{C}$ as a submodule of $\left(\mathbb{F}_q[x]\right)^\ell$ is generated by
\begin{equation}
\label{generating_set}
\begin{split}
&\left\{ \left( g_{i,1}(x) , g_{i,2}(x) ,\ldots, g_{i,\ell}(x)  \right) \big| 1\le i \le r \right\} \bigcup \\
&\qquad \qquad \left\{ \left(x^m-\lambda,0,\ldots,0\right),\left(0,x^m-\lambda,0,\ldots,0\right),\ldots, \left(0,\ldots,0,x^m-\lambda\right) \right\}.
\end{split}
\end{equation}

\begin{exmp}
We continue discussing the $(3,\theta)$-QT code $\mathcal{C}$ started in Example \ref{FirstQTCode}. Rows of the generator matrix given by \eqref{GMinExample} generate $\mathcal{C}$ as an $\mathbb{F}_4$-subspace of $\mathbb{F}_4^9$. Hence, they generate $\mathcal{C}$ as an $\mathbb{F}_4[x]$-submodule of $\mathbb{F}_4^9$, see the paragraph before Theorem \ref{QT_Corresp}. The polynomial representations of these rows are their image under the isomorphism $\phi$ given in the proof of Theorem \ref{QT_Corresp}; they construct the set 
\begin{equation*}
\begin{split}
&\left\{ \left( 1, (1+\theta)x^2, x^2 \right) ,  \left( (1+\theta)x^2, 1+\theta x^2, x^2 \right) ,  \left( x^2, 0, 1+\theta x^2 \right),\right. \\ 
&\left. \qquad\quad  \left( x+x^2, \theta x^2, \theta x^2\right) , \left(0, x+\theta x^2, x^2 \right)  ,  \left( (1+\theta) x^2, x^2, x+(1+\theta)x^2 \right)\right\}\subset \mathscr{R}_\theta^3.
\end{split}
\end{equation*}
Since $\phi$ is an $\mathbb{F}_4[x]$-module isomorphism, the above set generates $\mathcal{C}$ as an $\mathbb{F}_4[x]$-submodule of $\mathscr{R}_\theta^3$. From \eqref{generating_set}, a generating set for $\mathcal{C}$ as an $\mathbb{F}_4[x]$-submodule of $\left(\mathbb{F}_4[x]\right)^3$ is 
\begin{equation*}
\begin{split}
&\left\{ \left( 1, (1+\theta)x^2, x^2 \right) ,  \left( (1+\theta)x^2, 1+\theta x^2, x^2 \right) ,  \left( x^2, 0, 1+\theta x^2 \right), \right.\\ 
&\qquad\qquad \left( x+x^2, \theta x^2, \theta x^2\right) , \left(0, x+\theta x^2, x^2 \right)  ,  \left( (1+\theta) x^2, x^2, x+(1+\theta)x^2 \right),\\ 
&\left. \qquad\qquad\qquad\qquad\qquad\qquad\qquad\quad  \left( x^3+\theta, 0, 0\right) , \left(0, x^3+\theta, 0 \right)  ,  \left( 0, 0, x^3+\theta \right)\right\}.
\end{split}
\end{equation*}
This set is not necessarily a minimal generating set of $\mathcal{C}$, but we will show later how it can be reduced to a minimal one.\end{exmp}

Again, we will not distinguish between an $(\ell,\lambda)$-QT code $\mathcal{C}$ as a $T_{\ell,\lambda}$-invariant subspace of $\mathbb{F}_q^{m\ell}$ and the corresponding $\mathbb{F}_q[x]$-submodule of $\left(\mathbb{F}_q[x]\right)^\ell$ that contains $M$. That is, $\mathcal{C}$ is a linear code over $\mathbb{F}_q[x]$ of length $\ell$, where a linear code over a ring $R$ of length $\ell$ is an $R$-submodule of $R^\ell$. On the other hand, a linear code over $\mathbb{F}_q[x]$ of length $\ell$ is an $(\ell,\lambda)$-QT code if it contains $M$. However, we note that the minimum distance, code length, and rank of $\mathcal{C}$ as a linear code over $\mathbb{F}_q[x]$ are different from those of $\mathcal{C}$ as a linear code over $\mathbb{F}_q$. We aim to find a generator matrix for $\mathcal{C}$ as a linear code over $\mathbb{F}_q[x]$. We call any such matrix a GPM of $\mathcal{C}$ because its entries are polynomials over $\mathbb{F}_q$. A GPM can be constructed from the generating set given by \eqref{generating_set}. Since a GPM is a matrix over the PID $\mathbb{F}_q[x]$, one might ask for a unique reduced form GPM. Matrices over PIDs and their reduced form are discussed in Section \ref{Matrices_PIDs}.

We conclude this section by defining the class of MT codes that provides an additional generalization of QT codes. This generalization is based on generalizing the co-index into $\ell$ block lengths, which are not necessarily equal.
\begin{Def}
\label{def_MTcode}
Let $m_1, m_2, \ldots ,m_\ell$ be positive integers and $\Lambda=\left(\lambda_1, \lambda_2, \ldots, \lambda_\ell \right)$, where $0\ne\lambda_j \in\mathbb{F}_q$ for $1\le j\le \ell$. A $\Lambda$-MT code over $\mathbb{F}_q$ of code length $n=m_1+m_2+\cdots+m_\ell$ and block lengths $(m_1,m_2,\ldots,m_\ell)$ is an $\mathbb{F}_q[x]$-submodule of $\left(\mathbb{F}_q[x]\right)^\ell$ that contains the submodule
\begin{equation*}
\begin{split}
M_\Lambda=\left\{\big((x^{m_1}-\lambda_1)f_1(x), (x^{m_2}-\lambda_2)f_2(x), \ldots, (x^{m_\ell}-\lambda_\ell)f_\ell(x)\big)   \qquad \qquad \qquad\right. \\
\left. \text{such that } f_j(x)\in\mathbb{F}_q[x] \text{ for } 1\le j\le \ell\right\}.
\end{split}
\end{equation*}\end{Def}

From Theorem \ref{QT_Corresp2}, an $(\ell,\lambda)$-QT code is $\Lambda$-MT of equal block lengths and $\Lambda=(\lambda,\lambda,\ldots,\lambda)$. 
\begin{Thm}
\label{MT_Corresp}
There is a one-to-one correspondence between $\left(\lambda_1, \lambda_2, \ldots, \lambda_\ell \right)$-MT codes over $\mathbb{F}_q$ of block lengths $(m_1,m_2,\ldots,m_\ell)$ and $T_{\Lambda}$-invariant $\mathbb{F}_q$-subspaces of $\mathbb{F}_q^{n}$, where $n=m_1+m_2+\cdots+m_\ell$ and $T_{\Lambda}$ is the automorphism of $\mathbb{F}_q^n$ given by
\begin{equation}
\label{T_ell_Lambda}
\begin{split}
&T_{\Lambda}:\left( a_{0,1},\ldots, a_{m_1-1,1},a_{0,2},\ldots, a_{m_2-1,2},\ldots, a_{0,\ell},\ldots, a_{m_\ell-1,\ell}\right)\mapsto\\
&\left(\lambda_1 a_{m_1-1,1},a_{0,1},\ldots, a_{m_1-2,1}, \lambda_2 a_{m_2-1,2},a_{0,2},\ldots, a_{m_2-2,2},\ldots,\right.\\
&\left.\qquad\qquad\qquad\qquad\qquad\qquad\qquad\qquad\qquad\qquad \lambda_\ell a_{m_\ell-1,\ell},a_{0,\ell},\ldots, a_{m_\ell-2,\ell}\right).
\end{split}
\end{equation}
\end{Thm}
\begin{proof}
For $1\le j\le \ell$, let $\mathscr{R}_{m_j,\lambda_j}=\mathbb{F}_q[x]/\langle x^{m_j}-\lambda_j\rangle$ and $\pi_j:\mathbb{F}_q[x]\rightarrow\mathscr{R}_{m_j,\lambda_j}$ be the projection homomorphism. Then $\pi=\oplus_{j=1}^\ell \pi_j : \left(\mathbb{F}_q[x]\right)^\ell \rightarrow \oplus_{j=1}^\ell \mathscr{R}_{m_j,\lambda_j}$ is a surjective homomorphism with kernel $M_\Lambda$. Actually, $\pi$ defines a one-to-one correspondence between $\mathbb{F}_q[x]$-submodules of $\oplus_{j=1}^\ell \mathscr{R}_{m_j,\lambda_j}$ and $\mathbb{F}_q[x]$-submodules of $\left(\mathbb{F}_q[x]\right)^\ell$ that contain $M_\Lambda$ given in Definition \ref{def_MTcode}. Hence, there is a one-to-one correspondence between $\left(\lambda_1, \lambda_2, \ldots, \lambda_\ell \right)$-MT codes over $\mathbb{F}_q$ of block lengths $(m_1,m_2,\ldots,m_\ell)$ and submodules of $\oplus_{j=1}^\ell \mathscr{R}_{m_j,\lambda_j}$.

The automorphism $T_{\Lambda}$ makes $\mathbb{F}_q^{n}$ into an $\mathbb{F}_q[x]$-module, where the action of $x$ on $\mathbb{F}_q^{n}$ is the action of $T_{\Lambda}$. In this setting, $\mathbb{F}_q[x]$-submodules of $\mathbb{F}_q^{n}$ are precisely the $T_{\Lambda}$-invariant $\mathbb{F}_q$-subspaces of $\mathbb{F}_q^{n}$. We establish a one-to-one correspondence between submodules of $\oplus_{j=1}^\ell \mathscr{R}_{m_j,\lambda_j}$ and $T_{\Lambda}$-invariant $\mathbb{F}_q$-subspaces of $\mathbb{F}_q^{n}$ by obtaining an $\mathbb{F}_q[x]$-module isomorphism between $\oplus_{j=1}^\ell \mathscr{R}_{m_j,\lambda_j}$ and $\mathbb{F}_q^{n}$. There is an $\mathbb{F}_q$-vector space isomorphism $\phi: \mathbb{F}_q^n\rightarrow \oplus_{j=1}^\ell \mathscr{R}_{m_j,\lambda_j}$ given by
\begin{equation}
\label{the_isomorphism_phi}
\left( a_{0,1},\ldots, a_{m_1-1,1},\ldots, a_{0,\ell},\ldots, a_{m_\ell-1,\ell}\right) \mapsto \left( a_1(x),a_2(x),\ldots,a_\ell(x)\right),
\end{equation}
where $a_j(x)=a_{0,j}+a_{1,j}x+\cdots+ a_{m_j-1,j}x^{m_j-1}$ for $1\le j\le \ell$. Consider the following diagram of $\mathbb{F}_q$-vector space isomorphisms
\begin{equation}
\label{Commut_diagram2}
\begin{tikzcd}[swap]
    \mathbb{F}_q^{n} \arrow{r}[swap]{\phi}{} \arrow{d}{T_{\Lambda}}  & \oplus_{j=1}^\ell \mathscr{R}_{m_j,\lambda_j} \arrow{d}[swap]{\psi}{} \\  
    \mathbb{F}_q^{n} \arrow{r}[swap]{}{\phi}   & \oplus_{j=1}^\ell \mathscr{R}_{m_j,\lambda_j}
\end{tikzcd}
\end{equation}
where $\psi:  \left( a_1(x), a_2(x), \ldots, a_\ell(x)\right) \mapsto \left(x a_1(x),x a_2(x), \ldots,x a_\ell(x)\right)$. 
Similar to the proof of Theorem \ref{QT_Corresp}, Equation \eqref{T_ell_Lambda} shows that Diagram \eqref{Commut_diagram2} is commutative. Then $x\phi\left(\mathbf{a}\right)=\phi\left(T_{\Lambda}(\mathbf{a})\right)$ for any $\mathbf{a}\in\mathbb{F}_q^{n}$, and $\phi$ is an $\mathbb{F}_q[x]$-module isomorphism.

If $\mathcal{C}$ is a $\left(\lambda_1, \lambda_2, \ldots, \lambda_\ell \right)$-MT code over $\mathbb{F}_q$ of block lengths $(m_1,m_2,\ldots,m_\ell)$, then $\phi^{-1}\circ\pi\left(\mathcal{C}\right)$ is a $T_{\Lambda}$-invariant $\mathbb{F}_q$-subspace of $\mathbb{F}_q^{n}$. Conversely, the image of any $T_{\Lambda}$-invariant $\mathbb{F}_q$-subspace of $\mathbb{F}_q^{n}$ under the map $\pi^{-1}\circ\phi$ is an $\mathbb{F}_q[x]$-submodule of $\left(\mathbb{F}_q[x]\right)^\ell$ that contains $M_\Lambda$.
\end{proof}

In literature, a MT code may mean a $T_{\Lambda}$-invariant subspace of $\mathbb{F}_q^{n}$, a submodule of $\oplus_{j=1}^\ell \mathscr{R}_{m_j,\lambda_j}$, or a submodule of $\left(\mathbb{F}_q[x]\right)^\ell$ that contains $M_\Lambda$. We do not distinguish between the three representations, and the representation used is determined by context. However, the polynomial representation of a MT-code is the corresponding submodule of $\oplus_{j=1}^\ell \mathscr{R}_{m_j,\lambda_j}$.

Similar to what we did with QT codes, if a $\Lambda$-MT code $\mathcal{C}$ is generated as a submodule of $\oplus_{j=1}^\ell \mathscr{R}_{m_j,\lambda_j}$ by the set
\begin{equation*}
\left\{ \left( g_{i,1}(x)+\langle x^{m_1}-\lambda_1 \rangle ,g_{i,2}(x)+\langle x^{m_2}-\lambda_2 \rangle ,\ldots, g_{i,\ell}(x)+\langle x^{m_\ell}-\lambda_\ell \rangle \right) \big| 1\le i \le r \right\},
\end{equation*}
where $g_{i,j}(x)\in\mathbb{F}_q[x]$ for $1\le i\le r$ and $1\le j\le \ell$, then $\mathcal{C}$ is generated as a submodule of $\left(\mathbb{F}_q[x]\right)^\ell$ by the set 
\begin{equation*}
\begin{split}
&\left\{ \left( g_{i,1}(x) , g_{i,2}(x) ,\ldots, g_{i,\ell}(x)  \right) \big| 1\le i \le r \right\} \bigcup \\
&\qquad \qquad \left\{ \left(x^{m_1}-\lambda_1,0,\ldots,0\right),\left(0,x^{m_2}-\lambda_2,0,\ldots,0\right),\ldots, \left(0,\ldots,0,x^{m_\ell}-\lambda_\ell\right) \right\}.
\end{split}
\end{equation*}
Hence, $\mathcal{C}$ is a linear code over $\mathbb{F}_q[x]$ of length $\ell$ with generator matrix
\begin{equation*}
\mathbf{G}=\begin{pmatrix}
g_{1,1}(x) & g_{1,2}(x) & g_{1,3}(x) & \cdots & g_{1,\ell}(x)\\
g_{2,1}(x) & g_{2,2}(x) & g_{2,3}(x) & \cdots & g_{2,\ell}(x)\\
\vdots & \vdots & \vdots & \ddots & \vdots \\
g_{r,1}(x) & g_{r,2}(x) & g_{r,3}(x) & \cdots & g_{r,\ell}(x)\\
x^{m_1}-\lambda_1 & 0 & 0 & \cdots & 0\\
0 & x^{m_2}-\lambda_2 & 0 &  \cdots & 0\\
\vdots & \vdots & \vdots & \ddots & \vdots \\
0 & 0 & 0 & \cdots & x^{m_\ell}-\lambda_\ell
\end{pmatrix}.
\end{equation*}
Since the entries of $\mathbf{G}$ are polynomials, so it is called GPM. Rows of $\mathbf{G}$ are not necessary to form a minimal generating set for $\mathcal{C}$. However, a reduced form for $\mathbf{G}$ can be obtained, see the next section.

In the standard basis of $\mathbb{F}_q^n$, the matrix representation of $T_{\Lambda}$ is the block diagonal matrix 
\begin{equation}
\label{M_i_MT}
\mathbf{M}_\Lambda=\mathrm{diag}\left[\mathbf{M}_{\lambda_1,m_1},\mathbf{M}_{\lambda_2,m_2},\ldots, \mathbf{M}_{\lambda_\ell,m_\ell}\right], 
\end{equation}
where 
\begin{equation*}
\mathbf{M}_{\lambda_j,m_j}=\begin{pmatrix}
\mathbf{0}_{1\times(m_j-1)} & \lambda_j\\
\mathbf{I}_{(m_j-1)} & \mathbf{0}_{(m_j-1)\times 1}
\end{pmatrix}
\text{ for } 1\le j\le \ell.
\end{equation*} 
The following result can be proven in a similar way to the proof of Theorem \ref{Cond_QT}.
\begin{Thm}
\label{Cond_MT2}
Let $\mathcal{C}$ be a linear code over $\mathbb{F}_q$ of length $n$, dimension $k$, and a $k\times n$ generator matrix $\mathbf{G}$. Then, $\mathcal{C}$ is $\left(\lambda_1, \lambda_2, \ldots, \lambda_\ell \right)$-MT with block lengths $(m_1,m_2,\ldots,m_\ell)$ if and only if  
\begin{equation*}
\text{rank}\begin{pmatrix}
\mathbf{G}\\
\mathbf{G} \mathbf{M}_\Lambda^t
\end{pmatrix}=k.
\end{equation*}\end{Thm}

\begin{exmp}
\label{ex_Sh}
Let $\mathcal{C}$ be the linear code over $\mathbb{F}_3$ of length $60$, dimension $6$, and generator matrix
$\mathbf{G}=\begin{pmatrix}
\mathbf{I}_{6} & \mathbf{N} \end{pmatrix}$, 
where $\mathbf{N}$ is given by
\begin{equation*}
\resizebox{\hsize}{!}{
$\left(\begin{array}{*{54}c}
1&2&1&2&2&1&0&2&0&2&1&2&0&1&1&0&1&1&0&2&0&2&0&0&2&2&0&2&0&1&0&1&1&2&2&0&2&2&0&1&0&1&0&0&1&1 &0&1&0&2&0&2&2&1\\
1&0&0&0&1&0&1&2&2&2&0&0&2&1&2&1&1&2&1&2&2&2&2&0&2&1&2&2&2&1&1&1&2&0&1&2&2&1&2&1&1&1&1&0&1&2 &1&1&1&2&2&2&1&0\\
1&0&1&2&2&2&0&0&2&1&0&2&0&0&1&2&2&2&2&0&2&1&2&2&2&1&1&1&2&0&1&2&2&1&2&1&1&1&1&0&1&2&1&1&1&2 &2&2&1&0&2&1&1&2\\
0&1&0&1&2&2&2&0&0&2&1&0&2&0&2&1&2&2&2&2&0&2&1&2&2&2&1&1&1&2&0&1&2&2&1&2&1&1&1&1&0&1&2&1&1&1 &2&2&2&1&0&2&1&1\\
0&0&1&0&1&2&2&2&0&0&2&1&0&2&1&2&1&2&2&2&2&0&2&1&2&2&2&1&1&1&2&0&1&2&2&1&2&1&1&1&1&0&1&2&1&1 &1&2&2&2&1&0&2&1\\
2&1&2&2&1&0&2&0&2&1&2&0&1&2&0&1&1&0&2&0&2&0&0&2&2&0&2&0&1&0&1&1&2&2&0&2&2&0&1&0&1&0&0&1&1&0 &1&0&2&0&2&2&1&1
\end{array}
\right).$}
\end{equation*}
By brute-force, the minimum distance of $\mathcal{C}$ is $36$, hence $\mathcal{C}$ is optimal according to \cite{Grassl}. Furthermore, 
Theorem \ref{Cond_MT2} emphasizes that $\mathcal{C}$ is $(2,1)$-MT with block lengths $m_1=20$ and $m_2=40$ because 
\begin{equation*}
\text{rank}\begin{pmatrix}
\mathbf{G}\\
\mathbf{G} \mathbf{M}_\Lambda^t
\end{pmatrix}=6.
\end{equation*} 
A generating set $S$ for $\mathcal{C}$ as an $\mathbb{F}_3[x]$-submodule of $\mathscr{R}_{20,2} \oplus \mathscr{R}_{40,1}$ can be obtained from the polynomial representations of rows of $\mathbf{G}$. Namely, 
\begin{align*}
S=\Big\{ &\left(  1+x^{6}+2x^{7}+x^{8}+2x^{9}+2x^{10}+x^{11}+2x^{13}+2x^{15}+x^{16}+2x^{17}+x^{19},\right.\\
 &\quad 1+x^{2}+x^{3}+2x^{5}+2x^{7}+2x^{10}+2x^{11}+2x^{13}+x^{15}+x^{17}+x^{18}+2x^{19}+2x^{20}\\
&\left. \qquad\ \ +2x^{22}+2x^{23}+x^{25}+x^{27}+x^{30}+x^{31}+x^{33}+2x^{35}+2x^{37}+2x^{38}+x^{39} \right), \\
&\left(  x+x^{6}+x^{10}+x^{12}+2x^{13}+2x^{14}+2x^{15}+2x^{18}+x^{19},\right.\\
& \quad 2+x+x^{2}+2x^{3}+x^{4}+2x^{5}+2x^{6}+2x^{7}+2x^{8}+2x^{10}+x^{11}+2x^{12}+2x^{13}\\
& \quad +2x^{14}+x^{15}+x^{16}+x^{17}+2x^{18}+x^{20}+2x^{21}+2x^{22}+x^{23}+2x^{24}+x^{25}+x^{26}\\
& \left. \qquad\quad+x^{27}+x^{28}+x^{30}+2x^{31}+x^{32}+x^{33}+x^{34}+2x^{35}+2x^{36}+2x^{37}+x^{38} \right), \\
&\left(  x^2+x^{6}+x^{8}+2x^{9}+2x^{10}+2x^{11}+2x^{14}+x^{15}+2x^{17},\right.\\
& \quad 1+2x+2x^{2}+2x^{3}+2x^{4}+2x^{6}+x^{7}+2x^{8}+2x^{9}+2x^{10}+x^{11}+x^{12}+x^{13}\\
& \quad +2x^{14}+x^{16}+2x^{17}+2x^{18}+x^{19}+2x^{20}+x^{21}+x^{22}+x^{23}+x^{24}+x^{26}+2x^{27}\\
& \left. \qquad\ \ +x^{28}+x^{29}+x^{30}+2x^{31}+2x^{32}+2x^{33}+x^{34}+2x^{36}+x^{37}+x^{38}+2x^{39} \right), \\
&\left(  x^3+x^{7}+x^{9}+2x^{10}+2x^{11}+2x^{12}+2x^{15}+x^{16}+2x^{18},\right.\\
& \quad 2+x+2x^{2}+2x^{3}+2x^{4}+2x^{5}+2x^{7}+x^{8}+2x^{9}+2x^{10}+2x^{11}+x^{12}+x^{13}\\
& \quad +x^{14}+2x^{15}+x^{17}+2x^{18}+2x^{19}+x^{20}+2x^{21}+x^{22}+x^{23}+x^{24}+x^{25}+x^{27}\\
& \left. \qquad\ \ +2x^{28}+x^{29}+x^{30}+x^{31}+2x^{32}+2x^{33}+2x^{34}+x^{35}+2x^{37}+x^{38}+x^{39} \right), \\
&\left(  x^4+x^{8}+x^{10}+2x^{11}+2x^{12}+2x^{13}+2x^{16}+x^{17}+2x^{19},\right.\\
& \quad 1+2x+x^{2}+2x^{3}+2x^{4}+2x^{5}+2x^{6}+2x^{8}+x^{9}+2x^{10}+2x^{11}+2x^{12}+x^{13}\\
& \quad +x^{14}+x^{15}+2x^{16}+x^{18}+2x^{19}+2x^{20}+x^{21}+2x^{22}+x^{23}+x^{24}+x^{25}+x^{26}\\
&\left.  \qquad\ \ +x^{28}+2x^{29}+x^{30}+x^{31}+x^{32}+2x^{33}+2x^{34}+2x^{35}+x^{36}+2x^{38}+x^{39} \right), \\
&\left(  x^5+2x^{6}+x^{7}+2x^{8}+2x^{9}+x^{10}+2x^{12}+2x^{14}+x^{15}+2x^{16}+x^{18}+2x^{19},\right.\\
&\ \    x+x^{2}+2x^{4}+2x^{6}+2x^{9}+2x^{10}+2x^{12}+x^{14}+x^{16}+x^{17}+2x^{18}+2x^{19}+2x^{21}\\
&\left.  \qquad\ \ +2x^{22}+x^{24}+x^{26}+x^{29}+x^{30}+x^{32}+2x^{34}+2x^{36}+2x^{37}+x^{38}+x^{39} \right)\Big\}.
\end{align*}
However, as an $\mathbb{F}_3[x]$-submodule of $\left( \mathbb{F}_3[x]\right)^2$, $\mathcal{C}$ is generated by 
\begin{equation}
\label{in_ex_Sh}
S\cup \left\{(x^{20}+1,0),(0,x^{40}+2)\right\}. 
\end{equation}
As a linear code over $\mathbb{F}_3[x]$, a GPM for $\mathcal{C}$ can be constructed from the latter set. In Section \ref{Matrices_PIDs}, we will discuss how to reduce this generating set to a minimal set, see Example \ref{ex_Sh2} below.\end{exmp}

GQC codes form a subclass of MT codes; an $\ell$-GQC code of block lengths $(m_1,m_2,\ldots,m_\ell)$ is $\left(1,1,\ldots,1\right)$-MT. The following is a corollary of Theorem \ref{MT_Corresp}. 
\begin{Coroll}
\label{Corollary_GQC}
Any of the following can be used to represent an $\ell$-GQC code of block lengths $(m_1,m_2,\ldots,m_\ell)$ over $\mathbb{F}_q$:
\begin{enumerate}
\item An $\mathbb{F}_q[x]$-submodule of $\left(\mathbb{F}_q[x]\right)^\ell$ that contains 
\begin{equation*}
\left\{\left((x^{m_1}\!-1)f_1(x),\! (x^{m_2}\!-1)f_2(x), \ldots,\! (x^{m_\ell}\!-1)f_\ell(x)\right) \big| f_1(x),\ldots,f_\ell(x)\!\in\!\mathbb{F}_q[x]\right\}\!.
\end{equation*}
Hence, an $\ell$-GQC is a linear code over $\mathbb{F}_q[x]$ of length $\ell$.
\item An $\mathbb{F}_q[x]$-submodule of $\oplus_{j=1}^\ell \mathscr{R}_{m_j,1}$, where $\mathscr{R}_{m_j,1}=\mathbb{F}_q[x]/\langle x^{m_j}-1\rangle$.
\item An invariant $\mathbb{F}_q$-subspace of $\mathbb{F}_q^{n}$, where $n=\sum_{j=1}^\ell m_j$, under the automorphism
\begin{equation*}
\begin{split}
&\left( a_{0,1},\ldots, a_{m_1-1,1},a_{0,2},\ldots, a_{m_2-1,2},\ldots, a_{0,\ell},\ldots, a_{m_\ell-1,\ell}\right)\mapsto\\
&\quad\qquad\qquad\qquad\left( a_{m_1-1,1},a_{0,1},\ldots, a_{m_1-2,1},\ldots, a_{m_\ell-1,\ell},a_{0,\ell},\ldots, a_{m_\ell-2,\ell}\right).
\end{split}
\end{equation*} 
\end{enumerate}\end{Coroll}

In the standard basis of $\mathbb{F}_q^{n}$ and using \eqref{M_i_MT} with $\lambda_j=1$ for $1\le j\le \ell$, the matrix representation of the automorphism given in Corollary \ref{Corollary_GQC} is
\begin{equation*}
\mathbf{M}_{\mathbf{1}}=\mathrm{diag}\left[\mathbf{M}_{1,m_1},\mathbf{M}_{1,m_2},\ldots, \mathbf{M}_{1,m_\ell}\right].
\end{equation*}
For $1\le j\le \ell$, $\mathbf{M}_{1,m_j}$ is a permutation matrix of size $m_j$. In addition, $\mathbf{M}_{1,m_j}$ has multiplicative order $m_j$, i.e., $\mathbf{M}_{1,m_j}^{m_j}=\mathbf{I}_{m_j}$. Consequently, $\mathbf{M}_{\mathbf{1}}$ is a permutation matrix with multiplicative order $N=\mathrm{lcm}\left(m_1,m_2,\ldots, m_\ell \right)$, the least common multiple of $m_1,m_2,\ldots, m_\ell$. This can be generalized to the matrix $\mathbf{M}_\Lambda$ given by \eqref{M_i_MT}; the multiplicative order of $\mathbf{M}_\Lambda$ is $\mathrm{lcm}\left(t_1 m_1,t_2 m_2,\ldots, t_\ell m_\ell \right)$, where $t_j$ is the multiplicative order of $\lambda_j$ for $1\le j\le \ell$. Equivalently, $T_\Lambda^{\mathrm{lcm}\left(t_1 m_1,t_2 m_2,\ldots, t_\ell m_\ell \right)}$ is the identity map.

\section{GPMs as matrices over PID}
\label{Matrices_PIDs}
Let ${R}$ be a commutative ring with identity and $\mathcal{M}$ be a finitely generated ${R}$-module. The rank of $\mathcal{M}$ is the size of the minimal generating set of $\mathcal{M}$. An $R$-module $\mathcal{M}$ is free if there exists an ${R}$-linearly independent generating set $\{\nu_1, \nu_2,\ldots, \nu_r\}$ for $\mathcal{M}$. The set $\{\nu_1, \nu_2,\ldots, \nu_r\}$ forms a basis for $\mathcal{M}$. That is, for each $\mathrm{m}\in\mathcal{M}$, there is a unique $\left(a_1,a_2,\ldots,a_r \right)\in{R}^r$ such that $\mathrm{m}=\sum_{i=1}^r a_i \nu_i$. The ``\textit{invariant basis number}'' property for commutative rings asserts that any two bases of $\mathcal{M}$ have the same cardinality, and that a basis forms a minimal generating set. Hence, if $\{\nu_1, \nu_2,\ldots, \nu_r\}$ is a basis of an $R$-module $\mathcal{M}$, then $\mathcal{M}$ is isomorphic to ${R}^r$ and has rank $r$. The free $R$-module $R^r$ of rank $r$ has the standard basis $\{\mathbf{e}_1,\mathbf{e}_1,\ldots,\mathbf{e}_r\}$, where $\mathbf{e}_i$ is the vector that has $1$ in the $i^\mathrm{th}$ coordinate and zeros in the remaining coordinates. Many interesting properties of modules over PID can be found in \cite{Roman2008-wr}. In the following, we assume ${R}$ PID. The reason for the interest in modules over PIDs is that many of the known results of vector spaces still apply to these modules. We summarize some of these properties without proof.
\begin{Thm}
\label{finitely_modules}
Let $\mathcal{M}$ be a finitely generated $R$-module, where $R$ is PID.
\begin{enumerate}
\item If $\mathcal{M}$ is free of rank $r$, then $\mathcal{M}$ is torsion-free. That is, $a \mathrm{m}=0$ for $a\in{R}$ and $\mathrm{m}\in\mathcal{M}$ if and only if $a=0$ or $\mathrm{m}=0$. This is easy to see since $\mathcal{M}$ is isomorphic to $R^r$. Thus, if $a \mathrm{m}=0$ for $a\ne 0$, then $a \left(a_1,a_2,\ldots,a_r\right)=\mathbf{0}$, where $\left(a_1,a_2,\ldots,a_r\right)\in{R}^r$ is the image of $\mathrm{m}$ under the isomorphism $\mathcal{M}\simeq {R}^r$. Since ${R}$ is an integral domain, $a a_i=0$ and $a_i=0$ for every $1\le i\le r$. Hence $\left(a_1,a_2,\ldots,a_r\right)=0$ and $\mathrm{m}=0$.
\item If $\mathcal{M}$ is torsion-free, then $\mathcal{M}$ is free.
\item If $\mathcal{N}$ is a submodule of $\mathcal{M}$, then rank($\mathcal{N}$) $\le$ rank($\mathcal{M}$). 
\item If $\mathcal{M}$ is free and $\mathcal{N}$ is a submodule, then $\mathcal{N}$ is free.
\item Similar to vector spaces, if $\mathcal{M}$ is free of rank $r$, then any generating set of size $r$ forms a basis of $\mathcal{M}$. On the other hand, unlike vector spaces, if $\mathcal{M}$ is free of rank $r$, then an ${R}$-linearly independent set of size $r$ does not necessarily generate $\mathcal{M}$. For instance, let ${R}=\mathbb{Z}$, $\mathcal{M}=\mathbb{Z}^2$, and $S=\{\mathbf{e}_1, 2\mathbf{e}_2\}$. 
Although $S$ is $\mathbb{Z}$-linearly independent of size $2$, it does not generate $(0,1)$, i.e., $\langle S \rangle \ne \mathcal{M}$. 
\item Let $\mathcal{M}$ be free of rank $r$ and let $\mathcal{N}$ be a submodule of rank $t$. There is a subset $\{\tau_1,\tau_2,\ldots,\tau_t\}$ of some basis $\{\nu_1,\nu_2,\ldots,\nu_r\}$ of $\mathcal{M}$ and elements $s_1,s_2,\ldots,s_t\in{R}$ such that 
\begin{enumerate}
\item $s_1|s_2|s_3|\cdots|s_t$, and
\item $\{s_1\tau_1,s_2\tau_2,\ldots,s_t\tau_t\}$ is a basis of $\mathcal{N}$.

\end{enumerate}
The elements $s_1,s_2,\ldots,s_t$ are called the invariant factors of $\mathcal{N}$ and are independent of the choice of basis of $\mathcal{M}$.

\end{enumerate}\end{Thm}

Throughout this section, let $\mathcal{C}$ refer to a $\Lambda$-MT code over $\mathbb{F}_q$ with block lengths $\left(m_1,m_2,\ldots,m_\ell\right)$, where $\Lambda=\left(\lambda_1,\lambda_2,\ldots,\lambda_\ell\right)$. From Definition \ref{def_MTcode}, $\mathcal{C}$ is a submodule of a free module of rank $\ell$ over a PID. Hence, Theorem \ref{finitely_modules} asserts that $\mathcal{C}$ is a free $\mathbb{F}_q[x]$-module of rank at most $\ell$. We aim to find a basis of $\mathcal{C}$ in a reduced form. A generating set of $\mathcal{C}$ as a submodule of $\left(\mathbb{F}_q[x]\right)^\ell$ can be obtained from any generator matrix of $\mathcal{C}$ as a linear code over $\mathbb{F}_q$, see Example \ref{ex_Sh}. This generating set can be reduced to a basis of $\mathcal{C}$ using the Hermite normal form of matrices over PIDs. More generally, suppose that $R$ is a PID and that $\mathcal{M}$ is an ${R}$-submodule of ${R}^\ell$. From Theorem \ref{finitely_modules}, $\mathcal{M}$ is a finitely generated free module. Suppose that $S=\left\{\mathbf{g}_1,\mathbf{g}_2,\ldots,\mathbf{g}_r \right\} \subseteq\mathcal{M}$ is a generating set of $\mathcal{M}$, where $\mathbf{g}_i=\left( g_{i,1},g_{i,2},\ldots,g_{i,\ell}\right)\in{R}^\ell$ for $1\le i\le r$. An $r\times \ell$ generator matrix for $\mathcal{M}$ is constructed as follows:
\begin{equation}
\label{Initial_GPM}
\mathbf{G}=\begin{pmatrix}
g_{1,1} & g_{1,2} & \ldots & g_{1,\ell}\\
g_{2,1} & g_{2,2} & \ldots & g_{2,\ell}\\
\vdots & \vdots & \ddots & \vdots \\
g_{r,1} & g_{r,2} & \ldots & g_{r,\ell}
\end{pmatrix}.
\end{equation}
Let $\phi: {R}^r \rightarrow \mathcal{M}$ be the ${R}$-module homomorphism defined by $\phi\left(\mathbf{a} \right)=\mathbf{a}\mathbf{G}$ for every $\mathbf{a}\in{R}^r$. Since $S$ generates $\mathcal{M}$, $\phi$ is surjective. Suppose that $\mathcal{N}$ is a submodule of $R^r$ that has a generator matrix $\mathbf{G}'$ of size $t\times \ell$. On the one hand, $\mathcal{N}$ is a submodule of $\mathcal{M}$ if and only if $\mathbf{G}'=\mathbf{N}\mathbf{G}$ for some matrix $\mathbf{N}$ over ${R}$ of size $t\times r$. On the other hand, if $\mathbf{G}'=\mathbf{U}\mathbf{G}$ for an invertible matrix $\mathbf{U}$, then $\mathcal{M}=\mathcal{N}$, because $\mathbf{G}'=\mathbf{U}\mathbf{G}$ implies $\mathcal{N}\subseteq\mathcal{M}$, while $\mathbf{G}=\mathbf{U}^{-1}\mathbf{G}'$ implies $\mathcal{M}\subseteq\mathcal{N}$. 

\begin{Def}
\label{Two_MT_Codes}
Let $\mathbf{G}$ and $\mathbf{G}'$ be two matrices over ${R}$ of the same size. We say that $\mathbf{G}$ and $\mathbf{G}'$ are left equivalent if $\mathbf{G}'=\mathbf{U}\mathbf{G}$ for some invertible matrix $\mathbf{U}$. Thus, two left equivalent matrices over $R$ generate the same $R$-module.\end{Def}

An Euclidean domain is an integral domain in which we can perform Euclidean division. Examples of Euclidean domains include $\mathbb{Z}$, $\mathbb{F}_q[x]$, and the ring of formal power series over any field. Every ideal in a Euclidean domain is generated by a single element, so Euclidean domains are automatically PIDs. A matrix over a Euclidean domain is invertible if and only if it is a product of elementary matrices, where an elementary matrix is the identity matrix after performing one elementary row operation. In addition, left and right multiplications by elementary matrices correspond to performing elementary row and column operations, respectively. Thus, for matrices over a Euclidean domain, a left equivalent matrix can be obtained by elementary row operations. Furthermore, for a module $\mathcal{M}$ over a Euclidean domain with a generator matrix $\mathbf{G}$, applying some elementary row operations to $\mathbf{G}$ results in a left equivalent matrix, the latter being a generator matrix for $\mathcal{M}$ as well. We will be interested in the case where $\mathcal{M}$ is a MT-code over $\mathbb{F}_q$, hence $R=\mathbb{F}_q[x]$ and $\mathbf{G}$ is a GPM.

\begin{Thm}
\label{Free_MTcode}
Let $\mathcal{C}$ be a $\Lambda$-MT code over $\mathbb{F}_q$ of block lengths $\left(m_1,m_2,\ldots,m_\ell\right)$ that has an $\ell\times\ell$ GPM $\mathbf{G}$. As an $\mathbb{F}_q[x]$-submodule of $\left(\mathbb{F}_q[x]\right)^\ell$, $\mathcal{C}$ is free of rank $\ell$. Moreover, rows of $\mathbf{G}$ form a basis for $\mathcal{C}$.\end{Thm} 
\begin{proof}
{We know that $\left(\mathbb{F}_q[x]\right)^\ell$ is a free module of rank $\ell$ over a PID. From Theorem \ref{finitely_modules}, $\mathcal{C}$ is free of rank $\le \ell$. The submodule $M_\Lambda$ given in Definition \ref{def_MTcode} is free of rank $\ell$; $M_\Lambda$ has the basis $\left\{(x^{m_j}-\lambda_j) \mathbf{e}_j\right\}_{j=1}^\ell$. Then rank($\mathcal{C}$) $\ge \ell$ because $\mathcal{C}\supseteq M_\Lambda$. Thus, $\mathcal{C}$ is free of rank $\ell$. Rows of $\mathbf{G}$ form a generating set for $\mathcal{C}$ of size equal to the rank, and thus form a basis by Theorem \ref{finitely_modules}.}
\end{proof}

\begin{Thm}
\label{Containment}
Let $\mathcal{C}$ and $\mathcal{C}'$ be two $\Lambda$-MT codes of block lengths $\left(m_1,m_2,\ldots,m_\ell\right)$ with $\ell\times\ell$ GPMs $\mathbf{G}$ and $\mathbf{G}'$, respectively. Then, $\mathcal{C}'\subseteq\mathcal{C}$ if and only if $\mathbf{G}'=\mathbf{U}\mathbf{G}$ for some matrix $\mathbf{U}$, where $\mathbf{U}$ is invertible if and only if $\mathcal{C}'=\mathcal{C}$.\end{Thm}
\begin{proof}
{We have $\mathcal{C}'\subseteq\mathcal{C}$ if and only if $\mathbf{G}$ generates the rows of $\mathbf{G}'$ if and only if $\mathbf{G}'=\mathbf{U}\mathbf{G}$ for some matrix $\mathbf{U}$. 

Assume $\mathbf{G}'=\mathbf{U}\mathbf{G}$ for an invertible $\mathbf{U}$. Then $\mathcal{C}'\subseteq\mathcal{C}$, and $\mathcal{C}\subseteq\mathcal{C}'$ because $\mathbf{G}=\mathbf{U}^{-1}\mathbf{G}'$, hence $\mathcal{C}'=\mathcal{C}$. 
Conversely, assume that $\mathcal{C}'=\mathcal{C}$. From the first part of this theorem, there is a matrix $\mathbf{V}$ such that $\mathbf{G}=\mathbf{V}\mathbf{G}'$. Now $\mathbf{G}'=\mathbf{U}\mathbf{G}=\mathbf{U}\mathbf{V}\mathbf{G}'$. That is, $\left(\mathbf{U}\mathbf{V}-\mathbf{I}_{\ell}\right)\mathbf{G}'=\mathbf{0}$. From Theorem \ref{Free_MTcode}, rows of $\mathbf{G}'$ form a basis for $\mathcal{C}$, hence $\mathbf{U}\mathbf{V}-\mathbf{I}_{\ell}=\mathbf{0}$. Then $\mathbf{U}$ is invertible.}
\end{proof}

\begin{Thm}
\label{equivalent}
Let $\mathcal{C}$ be a $\Lambda$-MT code over $\mathbb{F}_q$ of block lengths $\left(m_1,m_2,\ldots,m_\ell\right)$ that has an $\ell\times\ell$ GPM $\mathbf{G}$. An $\ell\times \ell$ matrix $\mathbf{G}'$ is a GPM of $\mathcal{C}$ if and only if $\mathbf{G}$ and $\mathbf{G}'$ are left equivalent.\end{Thm}
\begin{proof}
{Let $\mathcal{C}'$ be the $\mathbb{F}_q[x]$-module generated by $\mathbf{G}'$. From Theorem \ref{Containment}, $\mathcal{C}' =\mathcal{C}$ if and only if $\mathbf{G}'=\mathbf{U}\mathbf{G}$ for an invertible matrix $\mathbf{U}$ if and only if $\mathbf{G}$ and $\mathbf{G}'$ are left equivalent.}
\end{proof}

From the discussion before Theorem \ref{Free_MTcode} and the fact that $\mathbb{F}_q[x]$ is a PID, the following is proven.
\begin{Coroll}
Let $\mathcal{C}$ be a $\Lambda$-MT code over $\mathbb{F}_q$ of block lengths $\left(m_1,m_2,\ldots,m_\ell\right)$ that has an $\ell\times\ell$ GPM $\mathbf{G}$. An $\ell\times \ell$ matrix $\mathbf{G}'$ is a GPM of $\mathcal{C}$ if and only if $\mathbf{G}'$ can be obtained from $\mathbf{G}$ by elementary row operations.\end{Coroll}

In the representation of $\Lambda$-MT codes by $\mathbb{F}_q[x]$-submodules of $\left(\mathbb{F}_q[x]\right)^\ell$ containing $M_\Lambda$, the zero code is represented by $M_\Lambda$. However, $M_\Lambda$ is a free module with basis $\left\{(x^{m_j}-\lambda_j) \mathbf{e}_j\right\}_{j=1}^\ell$. Hence, a GPM for $M_\Lambda$ is the diagonal matrix 
\begin{equation*}
\mathbf{D}=\mathrm{diag}\left[x^{m_1}-\lambda_1,x^{m_2}-\lambda_2,\ldots,x^{m_\ell}-\lambda_\ell\right]=\begin{pmatrix}
x^{m_1}-\lambda_1 & 0 & \ldots & 0\\
0 & x^{m_2}-\lambda_2 & \ldots & 0\\
\vdots & \vdots & \ddots & \vdots \\
0 & 0 & \ldots & x^{m_\ell}-\lambda_\ell
\end{pmatrix}.
\end{equation*}
For any $\Lambda$-MT code $\mathcal{C}$ with a GPM $\mathbf{G}$, the $\mathbb{F}_q$-linearity of $\mathcal{C}$ indicates that it contains the zero code. Then from Theorem \ref{Containment}, there is a matrix $\mathbf{A}$ such that
\begin{equation}
\label{identical_eq}
\mathbf{A}\mathbf{G}=\mathbf{D}.
\end{equation}
Equation \eqref{identical_eq} is called the identical equation \cite{Matsui2015}. In particular, if $\mathcal{C}$ is $(\ell,\lambda)$-QT, then $\mathbf{A}$ and $\mathbf{G}$ commute.
\begin{Thm}
\label{A_G_Commutativity}
Let $\mathcal{C}$ be an $(\ell,\lambda)$-QT code with a GPM $\mathbf{G}$ that satisfies \eqref{identical_eq} by the matrix $\mathbf{A}$. Then $\mathbf{G}\mathbf{A}=\mathbf{D}$.\end{Thm}
\begin{proof}
{Assume $\mathbf{G}\mathbf{A}=\mathbf{B}$. Then $\mathbf{B}\mathbf{G}=\mathbf{G}\mathbf{A}\mathbf{G}=\mathbf{G}\mathbf{D}=\left(x^m-\lambda\right)\mathbf{G}=\mathbf{D}\mathbf{G}$. That is, $\left(\mathbf{B}-\mathbf{D}\right)\mathbf{G}=\mathbf{0}$. Since rows of $\mathbf{G}$ are $\mathbb{F}_q[x]$-linearly independent, $\left(\mathbf{B}-\mathbf{D}\right)=\mathbf{0}$ and $\mathbf{B}=\mathbf{D}$.}
\end{proof}

The matrix $\mathbf{A}$ plays a fundamental role in constructing a GPM for the dual code of a MT code. An $\ell\times\ell$ polynomial matrix over $\mathbb{F}_q[x]$ is GPM of a $\left(\lambda_1,\lambda_2,\ldots,\lambda_\ell\right)$-MT code if and only if it satisfies the identical equation. Let us fix $\ell$, $\Lambda$, and the block lengths $\left(m_1,m_2,\ldots,m_\ell\right)$. Set
\begin{equation*}
\mathscr{S}=\left\{\mathbf{G}_i \  \big|\  \exists\  \mathbf{A}_i \text{ such that } \mathbf{A}_i\mathbf{G}_i=\mathbf{D}\right\},
\end{equation*}
the set of all the $\ell\times\ell$ matrices that satisfy the identical equation. Then for any $\Lambda$-MT code with a GPM $\mathbf{G}$, we have $\mathbf{G}\in\mathscr{S}$. However, each $\mathbf{G}\in\mathscr{S}$ is a GPM for some $\Lambda$-MT code. We partition $\mathscr{S}$ with an equivalence relation $\sim$ defined by the left equivalence of matrices. That is, every pair $\left(\mathbf{G}_i, \mathbf{G}_j\right)$ of elements in $\mathscr{S}$ is equivalent, written $\mathbf{G}_i\sim\mathbf{G}_j$, if and only if $\mathbf{G}_i=\mathbf{U}\mathbf{G}_j$ for an invertible $\mathbf{U}$. The quotient $\mathscr{S}/\sim$ is the set of all equivalence classes of $\mathscr{S}$ by $\sim$. Theorem \ref{equivalent} shows a one-to-one correspondence between the set of all $\Lambda$-MT codes and $\mathscr{S}/\sim$. Our next goal is to find a unique simple representative for each equivalence class that can be obtained in an algorithmic way. The representative of each equivalence class is called the reduced GPM of the corresponding $\Lambda$-MT code. GPMs are matrices over the Euclidean domain $\mathbb{F}_q[x]$, hence the Hermite normal form of left equivalent matrices over a PID is an intuitive representation to use. In addition, a matrix over a PID is left equivalent to a unique Hermite normal form \cite{gathen_gerhard_2013}.

Two elements $r_1$ and $r_2$ in a PID ${R}$ are associate if there is a unit $u\in{R}$ such that $r_1=u r_2$. A complete set of non-associates of ${R}$ is a set $\mathscr{P}$ such that for every $r\in{R}$, there exists a unique $a\in\mathscr{P}$ associate of $r$. Examples of a complete set of non-associates include 
\begin{enumerate}
\item $\mathscr{P}=\{0,1,2,3,\ldots\}$ for ${R}=\mathbb{Z}$,
\item $\mathscr{P}$ is the set of all monic polynomials over $\mathbb{F}_q$ for ${R}=\mathbb{F}_q[x]$, and
\item $\mathscr{P}=\{x^i+x^{i+1}f(x)|f(x)\in{R}\text{ and }i\ge 0\}$ for ${R}=\mathbb{F}[\![x]\!]$, the ring of formal power series in $x$ over the field $\mathbb{F}$.
\end{enumerate}
If $R$ is a PID with a complete set of non-associates $\mathscr{P}$, then there is a one-to-one correspondence between the set of ideals of ${R}$ and $\mathscr{P}$. Specifically, the ideal $\langle r \rangle\subseteq R$ corresponds to the unique associate element of $r$ in $\mathscr{P}$. For every $a\in\mathscr{P}$, let $\mathscr{Q}(a)\subseteq{R}$ be a complete set of coset representatives of the quotient ring ${R}/\langle a \rangle$. That is, for each $b+\langle a \rangle \in {R}/\langle a \rangle$, there is a unique $c \in \mathscr{Q}(a)$ such that $c-b \in\langle a \rangle$. Continuing the above examples, we have
\begin{enumerate}
\item $\mathscr{Q}(a)=\{0,1,2,\ldots,a-1\}$ for ${R}=\mathbb{Z}$,
\item $\mathscr{Q}(a)$ is the set of all polynomials over $\mathbb{F}_q$ of degree less than that of $a$, and
\item $\mathscr{Q}(a)=\{a_0+a_1 x +\cdots+a_{i-1}x^{i-1}| a_0,a_1,\ldots,a_{i-1}\in\mathbb{F}\}$ for ${R}=\mathbb{F}[\![x]\!]$ and $a$ is an element of $\mathscr{P}$ in which the leading term is $x^{i}$.
\end{enumerate}
\begin{Def}
\label{Hermite_form}
Let ${R}$ be a PID, let $\mathscr{P}$ be a complete set of non-associates of $R$, and let $\mathscr{Q}(a)$ be a complete set of coset representatives for every $a\in\mathscr{P}$. A matrix $\mathbf{G}=[g_{i,j}]$ over $R$ of size $h\times \ell$ and rank $r$ is in Hermite normal form if the $i^\mathrm{th}$ row of $\mathbf{G}$ is zero for $r<i\le h$ and there is a sequence $1\le j_1 < j_2 < j_3 < \cdots < j_r \le \ell$ such that for every $1\le i\le r$
\begin{enumerate}
\item $0\ne g_{i,j_i}\in\mathscr{P}$, 
\item $g_{i,t}=0$ for all $1\le t<j_i$, and
\item $g_{t,j_i}\in\mathscr{Q}(g_{i,j_i})$ for all $1\le t<i$.

\end{enumerate}\end{Def}

As a typical example, for $R=\mathbb{Z}$, $\mathscr{P}=\{0,1,2,\ldots\}$, and $\mathscr{Q}(a)=\{0,1,2,\ldots,a-1\}$, the matrix 
\begin{equation*}
\begin{pmatrix}
0 & 3 & -5 & 4 &10 & 0 & 0\\
0 & 0 & 0  & 7 & -1 & 9 & 0\\
0 & 0 & 0  & 0 &  0 & 0 & 1\\
0 & 0 & 0  & 0 &  0 & 0 & 0
\end{pmatrix}
\end{equation*}
is in Hermite normal form. The Hermite normal form of any matrix $\mathbf{G}$ over a PID is the unique matrix that has the form described in Definition \ref{Hermite_form} and is left equivalent to $\mathbf{G}$. The Hermite normal form is unique as long as we fix a complete set of non-associates $\mathscr{P}$ and a complete set of coset representatives for each element of  $\mathscr{P}$. In literature, the Hermite normal form is the known reduced row echelon form when ${R}$ is a field. When $R$ is a Euclidean domain, the Hermite normal form can be obtained by applying elementary row operations.

Let $\mathcal{C}$ be a $\Lambda$-MT code over $\mathbb{F}_q$ and let $S$ be a generating set for $\mathcal{C}$ as an $\mathbb{F}_q[x]$-submodule of $\left(\mathbb{F}_q[x] \right)^\ell$. From \eqref{Initial_GPM}, $S$ can be used to construct a GPM for $\mathcal{C}$ of size $|S|\times \ell$. From now on, we shall use $\mathbf{G}$ to refer to the Hermite normal form of any GPM of $\mathcal{C}$ (see Theorem \ref{equivalent}), and we call $\mathbf{G}$ the reduced GPM of $\mathcal{C}$. From Theorem \ref{Free_MTcode}, $\mathcal{C}$ and $\mathbf{G}$ are of rank $\ell$. Hence, $\mathbf{G}$ is an $\ell\times\ell$ matrix of rank $\ell$ in Hermite normal form. Then Definition \ref{Hermite_form} shows that $\mathbf{G}$ is an upper triangular matrix having diagonal entries as non-zero monic polynomials. Let $\mathbf{A}=[a_{i,j}]$ be the matrix that satisfies the identical equation \eqref{identical_eq} of $\mathbf{G}$. Actually, $\mathbf{A}$ is upper triangular because $R$ is an integral domain and $\mathbf{G}$ and $\mathbf{D}$ are upper triangular matrices. Now, the identical equation shows that $a_{i,i}g_{i,i}=x^{m_i}-\lambda_i$ for $1\le i\le \ell$. That is, the $i^\mathrm{th}$ diagonal entry $g_{i,i}$ of $\mathbf{G}$ divides $\left(x^{m_i}-\lambda_i\right)$ for $1\le i\le \ell$.

\begin{exmp}
\label{ex_Sh2}
We continue with the $(2,1)$-MT code $\mathcal{C}$ given in Example \ref{ex_Sh}. We have shown that $\mathcal{C}$ as an $\mathbb{F}_3[x]$ submodule of $\left( \mathbb{F}_3[x]\right)^2$ is generated from the set given by \eqref{in_ex_Sh}. In \eqref{Initial_GPM}, we showed how such a generating set could build a GPM $\mathbf{M}$ for $\mathcal{C}$. Although the size of $\mathbf{M}$ is $8\times 2$, Theorem \ref{Free_MTcode} asserts that $\mathrm{rank}\left(\mathbf{M}\right)=2$. Since $\mathbf{M}$ has entries in the Euclidean domain $\mathbb{F}_3[x]$, elementary row operations are applied to $\mathbf{M}$ to get the reduced GPM $\mathbf{G}$ of $\mathcal{C}$. Actually, $\mathbf{G}$ is formed from the non-zero rows of the Hermite normal form of $\mathbf{M}$. Namely, 
\begin{equation*}
\mathbf{G}=\begin{pmatrix}
g_{1,1} & g_{1,2}\\
0 & x^{40}+2
\end{pmatrix},
\end{equation*}
where $g_{1,1}= 2+ x+2 x^2+ x^3+ x^4+2 x^5+ x^7+ x^9+2 x^{10}+ x^{11}+2 x^{13}+x^{14}$ and $g_{1,2}= x+ x^4+ x^5+ x^7+2 x^9+2 x^{11}+2 x^{12}+ x^{13}+ x^{14}+ x^{16}+ x^{17}+2 x^{19}+2 x^{21}+2 x^{24}+2 x^{25}+2 x^{27}+ x^{29}+ x^{31}+ x^{32}+2 x^{33}+2 x^{34}+2 x^{36}+2 x^{37}+x^{39}$. One can check that $\mathbf{G}$ satisfies the identical equation for
\begin{equation*}
\mathbf{A}=\begin{pmatrix}
2+2 x + x^4 + x^5+x^6 & \quad 2x(1+x)^4\\
0 & 1
\end{pmatrix}.
\end{equation*}\end{exmp}

\begin{Thm}
\label{dim_MT}
Let $\mathcal{C}$ be a $\Lambda$-MT code over $\mathbb{F}_q$ of block lengths $\left(m_1,m_2,\ldots,m_\ell\right)$ and let $\mathbf{G}=[g_{i,j}]$ be the reduced GPM of $\mathcal{C}$. Then $\mathcal{C}$ is an $\mathbb{F}_q$-vector space of dimension
\begin{equation*}
k=\sum_{j=1}^{\ell}\left( m_j-\mathrm{deg}(g_{j,j})\right).
\end{equation*}\end{Thm}
\begin{proof}
{Set $n=\sum_{j=1}^\ell m_j$. In this proof, we consider $\mathbb{F}_q^n$, $\left(\mathbb{F}_q[x]\right)^\ell$, and $\oplus_{j=1}^\ell \mathscr{R}_{m_j,\lambda_j}$ as $\mathbb{F}_q$-vector spaces. To remove any ambiguity, we write $\mathcal{C}$, $\mathcal{C}'$, and $\mathcal{C}''$ to refer to the code as a subspace of $\mathbb{F}_q^n$, $\left(\mathbb{F}_q[x]\right)^\ell$, and $\oplus_{j=1}^\ell \mathscr{R}_{m_j,\lambda_j}$, respectively. 

Let $\tau: \left(\mathbb{F}_q[x]\right)^\ell \rightarrow \mathcal{C}'$ be the map $\tau\left(\mathbf{a}\right) = \mathbf{a}\mathbf{G}$. Since the rows of $\mathbf{G}$ form a basis of $\mathcal{C}'$, $\tau$ is a vector space isomorphism. Let $\pi:\mathcal{C}' \rightarrow \mathcal{C}''$ be the restriction of the projection homomorphism $\left(\mathbb{F}_q[x]\right)^\ell \rightarrow \oplus_{j=1}^\ell \mathscr{R}_{m_j,\lambda_j}$, see the proof of Theorem \ref{MT_Corresp}. Then $\mathrm{Ker}\left(\pi\right)=\mathcal{C}' \cap M_\Lambda =M_\Lambda$. Let $\mathbf{A}=[a_{i,j}]$ be the matrix that satisfies the identical equation of $\mathbf{G}$, and let $\mathcal{A}$ be the $\mathbb{F}_q[x]$-submodule of $\left(\mathbb{F}_q[x]\right)^\ell$ generated by the rows of $\mathbf{A}$. We let $\mathcal{A}$ be viewed as an $\mathbb{F}_q$-vector space.

We claim that $\mathrm{Ker}\left(\pi\circ \tau\right)=\mathcal{A}$. To see this, we have $\mathcal{A} \subseteq \mathrm{Ker}\left(\pi \circ \tau\right)$ because $\pi \circ \tau \left( \mathbf{b}\mathbf{A}\right)=\pi \left( \mathbf{b}\mathbf{A}\mathbf{G}\right)=\pi \left( \mathbf{b}\mathbf{D}\right)=\mathbf{0}$ for any $\mathbf{b}\in \left(\mathbb{F}_q[x] \right)^\ell$. Conversely, if $\mathbf{a}\in \mathrm{Ker}\left(\pi \circ \tau\right)$, then $\tau\left( \mathbf{a}\right)\in \mathrm{Ker}\left( \pi\right)=M_\Lambda$, hence $\mathbf{a}\mathbf{G}=\mathbf{b}\mathbf{D}=\mathbf{b}\mathbf{A}\mathbf{G}$ for some $\mathbf{b}\in \left(\mathbb{F}_q[x] \right)^\ell$. Then, $\mathbf{a}=\mathbf{b}\mathbf{A}\in\mathcal{A}$ because the rows of $\mathbf{G}$ are $\mathbb{F}_q[x]$-linearly independent.

The isomorphism $\phi$ given by \eqref{the_isomorphism_phi} confirms that $\mathcal{C}\simeq\mathcal{C}''$. Now, the sequence of vector space homomorphisms
\begin{equation*}
\left(\mathbb{F}_q[x]\right)^\ell \xrightarrow[]{\tau} \mathcal{C}' \xrightarrow[]{\pi} \mathcal{C}'' \simeq \mathcal{C}.
\end{equation*}
shows that $\mathcal{C}\simeq \left(\mathbb{F}_q[x]\right)^\ell\!\!/\mathcal{A}$. It remains only to determine the dimension of the vector space $\left(\mathbb{F}_q[x]\right)^\ell\!\!/\mathcal{A}$.

Let $\mathbf{f}=\left( f_1(x),f_2(x),\ldots,f_\ell(x)\right)\in \left(\mathbb{F}_q[x]\right)^\ell$. Using the division algorithm iteratively, one can uniquely determine $q_j, r_j\in\mathbb{F}_q[x]$ for $1\le j\le \ell$ such that
\begin{equation*}
f_j(x)-\sum_{i=1}^{j-1} q_i a_{i,j} =q_j a_{j,j} +r_j,
\end{equation*}
where $\mathrm{deg}(r_j)< \mathrm{deg}(a_{j,j})=m_j-\mathrm{deg}(g_{j,j})$. Then 
\begin{equation*}
\mathbf{f}=[q_1,q_2,\ldots,q_\ell]\mathbf{A}+\mathbf{r},
\end{equation*} 
where $\mathbf{r}=\left( r_1,r_2,\ldots,r_\ell\right)$ and, hence, $\mathbf{f}-\mathbf{r}\in\mathcal{A}$. Therefore 
\begin{equation*}
\left(\mathbb{F}_q[x]\right)^\ell \!\!/ \mathcal{A}=\left\{( r_1,r_2,\ldots,r_\ell) + \mathcal{A} \ |\  \mathrm{deg}(r_j)< m_j-\mathrm{ deg}(g_{j,j}) \text{ for } 1\le j\le \ell\right\}.
\end{equation*}
Thus $\left(\mathbb{F}_q[x]\right)^\ell \!\!/ \mathcal{A}$ as a vector space has a dimension of $\sum_{j=1}^{\ell}\left( m_j-\mathrm{deg}(g_{j,j})\right)$.}
\end{proof}

\section{Dual codes of MT codes}
The standard inner product on $\mathbb{F}_q^n$ is defined by
\[\langle\mathbf{a},\mathbf{b}\rangle =\sum_{i=0}^{n-1} a_i b_i,\] 
where $\mathbf{a}=\left(a_0,a_1,\ldots,a_{n-1}\right)$ and $\mathbf{b}=\left(b_0,b_1,\ldots,b_{n-1}\right)$. Let $\mathcal{C}$ be a subset of $\mathbb{F}_q^n$, the dual of $\mathcal{C}$ is defined to be
\begin{equation}
\label{dual_general}
\mathcal{C}^\perp=\left\{\mathbf{a}\in\mathbb{F}_q^n \ |\ \langle \mathbf{a},\mathbf{c}\rangle =0 \ \forall \ \mathbf{c}\in\mathcal{C} \right\}.
\end{equation}
The inner product is a symmetric bilinear form, hence $\mathcal{C}^\perp$ is a linear code over $\mathbb{F}_q$ even if $\mathcal{C}$ is non-linear. Equation \eqref{dual_general} shows that $\mathcal{C}\subseteq\left(\mathcal{C}^\perp\right)^\perp$, where equality holds if $\mathcal{C}$ is linear. 

Let $\mathcal{C}$ be a linear code over $\mathbb{F}_q$ of length $n$, dimension $k$, and a $k\times n$ generator matrix $\mathbf{G}$. Then $\mathcal{C}^\perp$ is the null space of $\mathbf{G}$, or, equivalently, $\mathcal{C}^\perp=\{\mathbf{a}\in\mathbb{F}_q^n \ |\ \mathbf{G}\mathbf{a}^t =\mathbf{0}_{k\times 1}\}$. The rank plus nullity theorem \cite{Roman2008-wr} emphasizes that $\mathcal{C}^\perp$ is a linear subspace of $\mathbb{F}_q^n$ of dimension $n-k$. A code $\mathcal{C}$ is self-orthogonal if $\mathcal{C}^\perp\supseteq \mathcal{C}$, and in this case $n-k\ge k$, hence $k\le\frac{n}{2}$. However, $\mathcal{C}$ is self-dual if $\mathcal{C}^\perp=\mathcal{C}$. In fact, the code length of any self-dual code must be even because self-duality implies $n-k=k$, hence $n=2k$. Additionally, a self-dual code is precisely a self-orthogonal code whose dimension is equal to half the code length, i.e., $k=\frac{n}{2}$.

The MacWilliams identity is a fundamental result in coding theory that relates the weight enumerator of a linear code to the weight enumerator of its dual. Before we present MacWilliams identity in Theorem \ref{MacWilliams}, we define the weight enumerator of a code.
\begin{Def}
Let $\mathcal{C}$ be a code over $\mathbb{F}_q$ of length $n$ and let the weight of a codeword $\mathbf{c}\in\mathcal{C}$ be denoted by $\omega(\mathbf{c})$. The weight enumerator of $\mathcal{C}$ is a polynomial $W_\mathcal{C}(x,y)\in\mathbb{Z}[x,y]$ defined by 
\begin{equation*}
W_\mathcal{C}(x,y)=\sum_{\mathbf{c}\in\mathcal{C}}x^{n-\omega(\mathbf{c})}y^{\omega(\mathbf{c})}.
\end{equation*}\end{Def}

The size of the code can be determined from its weight enumerator since $W_\mathcal{C}(1,1)=|\mathcal{C}|$. Also, the polynomial $W_\mathcal{C}(1,y)$ describes the different weights of codewords in $\mathcal{C}$. For example, if $W_\mathcal{C}(1,y)=1+100 y^5+200 y^6+300 y^7$, then $\mathcal{C}$ contains the zero codeword, $100$ codewords of weight $5$, $200$ codewords of weight $6$, and $300$ codewords of weight $7$. Hence, $\mathcal{C}$ contains $601$ codewords and has a minimum weight of $5$. The MacWilliams identity relates $W_{\mathcal{C}}$ to $W_{\mathcal{C}^\perp}$. The identity is valid for linear codes over a wide class of rings called Frobenius rings \cite{Dougherty2017}. We state MacWilliams identity in its general form, where codes over finite fields are a special case.
\begin{Thm}[MacWilliams Identity]
\label{MacWilliams}
Let $\mathcal{C}$ be a linear code over a finite commutative Frobenius ring of size $q$. Then
\begin{equation*}
W_{\mathcal{C}^\perp}(x,y)=\frac{1}{|\mathcal{C}|}W_{\mathcal{C}}\left(x+(q-1)y,x-y\right).
\end{equation*}\end{Thm}

\begin{exmp}
\label{ex_Sh3}
We continue with the $(2,1)$-MT code $\mathcal{C}$ discussed in Examples \ref{ex_Sh} and \ref{ex_Sh2}. A brute force calculation of the weight enumerator of $\mathcal{C}$ shows that $W_\mathcal{C}(x,y)=x^{60} + 400 x^{24} y^{36} + 328 x^{15} y^{45}$. From Theorem \ref{MacWilliams}, we find that the weight enumerator of the dual of $\mathcal{C}$ is as follows:
\begin{equation*}
\begin{split}
W_{\mathcal{C}^\perp}(x,y) &=  \frac{1}{729}W_{\mathcal{C}}\left(x+2y,x-y\right)
\\ &=  x^{60}+40 x^{58}y^2 + 240 x^{57}y^3 + 8760 x^{56}y^4  + \cdots + \\
&\qquad\qquad\qquad   47445329187307520 x y^{59} + 1581510989447168 y^{60}.
\end{split}
\end{equation*}
Therefore, $\mathcal{C}^\perp$ is a linear code over $\mathbb{F}_3$ of length $60$, dimension $54$, and $d_{\mathrm{min}}\left(\mathcal{C}^\perp\right)=2$.\end{exmp}

We now turn our attention to MT codes and focus on investigating their duals. The dual of a MT code is not only a linear code, it is MT as well, but with different shift constants.
\begin{Thm}
\label{MT_dual}
Let $\mathcal{C}$ be a $\Lambda$-MT code over $\mathbb{F}_q$ of block lengths $\left(m_1,m_2,\ldots,m_\ell\right)$, where $\Lambda=\left(\lambda_1,\lambda_2,\ldots,\lambda_\ell \right)$. Then $\mathcal{C}^\perp$ is $\Delta$-MT of block lengths $\left(m_1,m_2,\ldots,m_\ell\right)$, where $\Delta=\left(\frac{1}{\lambda_1},\frac{1}{\lambda_2},\ldots,\frac{1}{\lambda_\ell} \right)$.\end{Thm}
\begin{proof}
{From Theorem \ref{MT_Corresp}, $\mathcal{C}$ is a $T_{\Lambda}$-invariant subspace of $\mathbb{F}_q^{n}$, where $n=\sum_{j=1}^\ell m_j$ and $T_{\Lambda}$ is the automorphism given by \eqref{T_ell_Lambda}. Let $N=\mathrm{lcm}\left(t_1 m_1,t_2 m_2,\ldots, t_\ell m_\ell \right)$, where $t_j$ is the multiplicative order of $\lambda_j$ for $1\le j\le \ell$. According to the discussion after Corollary \ref{Corollary_GQC}, $T_{\Lambda}^N$ is the identity map. We show that $T_{\Delta}\left(\mathbf{b}\right)\in\mathcal{C}^\perp$ for every $\mathbf{b}\in\mathcal{C}^\perp$. For any $\mathbf{b}\in\mathcal{C}^\perp$ and $\mathbf{c}\in\mathcal{C}$, we have
\begin{equation*}
\resizebox{\hsize}{!}{
$\langle\mathbf{c}, T_{\Delta}\left(\mathbf{b}\right)\rangle
=\langle T_{\Lambda}^N\left(\mathbf{c}\right), T_{\Delta}\left(\mathbf{b}\right)\rangle
=\langle T_{\Lambda}\circ T_{\Lambda}^{N-1}\left(\mathbf{c}\right), T_{\Delta}\left(\mathbf{b}\right)\rangle
=\langle T_{\Lambda}^{N-1}\left(\mathbf{c}\right), \mathbf{b}\rangle
=0$}
\end{equation*}
because $T_{\Lambda}^{N-1}\left(\mathbf{c}\right)\in\mathcal{C}$. Then, $T_{\Delta}\left(\mathbf{b}\right)\in\mathcal{C}^\perp$ and $\mathcal{C}^\perp$ is $\Delta$-MT.}
\end{proof}

Since the dual of a $\Lambda$-MT code is $\Delta$-MT, it can be represented by a submodule of $\left(\mathbb{F}_q[x] \right)^\ell$ containing $M_\Delta$, see Definition \ref{def_MTcode}. We aim to present a GPM for $\mathcal{C}^\perp$. The following definition provides such a GPM.
\begin{Def}
\label{def_parity}
For $1\le i\le \ell$, let $0\ne\lambda_i\in\mathbb{F}_q$ and let $m_i$ be a positive integer. For a MT code of index $\ell$, let $\mathbf{G}=[g_{i,j}]$ be the reduced GPM, let $\mathbf{A}=[a_{i,j}]$ be the matrix that satisfies the identical equation of $\mathbf{G}$, and let $d_j=\deg(g_{j,j})$ for $1\le j\le \ell$. 
\begin{enumerate}
\item Let $\mathbf{A}\left( \frac{1}{x}\right)$ be the matrix obtained from $\mathbf{A}$ when $x$ is replaced by $\frac{1}{x}$.
\item Let $\mathbf{A}^*$ be the matrix obtained after multiplying the $(i,j)$-th entry of $\mathbf{A}\left( \frac{1}{x}\right)$ by $x^{m_i-d_j}$.
\item (Eliminate the negative exponents in $\mathbf{A}^*$) Let $\mathbf{A}^{**}$ be the matrix obtained from $\mathbf{A}^*$ by reducing the $(i,j)$-th entry (for $i<j$) of $\mathbf{A}^*$ modulo $\left(x^{m_i}-\frac{1}{\lambda_i}\right)$. Specifically, $x^{-\mu}$ is replaced by $\lambda_i x^{m_i-\mu}$ for $\mu\ge 1$.
\item The matrix $\mathbf{H}$ is defined to be the transpose of $\mathbf{A}^{**}$.

\end{enumerate}\end{Def}

\begin{Thm}
\label{MT_Parity_H}
Let $\mathcal{C}$ be a $\left(\lambda_1,\lambda_2,\ldots,\lambda_\ell \right)$-MT code over $\mathbb{F}_q$ of block lengths $\left(m_1,m_2,\ldots,m_\ell\right)$, let $\mathbf{G}=[g_{i,j}]$ be the reduced GPM of $\mathcal{C}$, and let $\mathbf{A}=[a_{i,j}]$ be the matrix that satisfies the identical equation of $\mathbf{G}$. The polynomial matrix $\mathbf{H}$ given in Definition \ref{def_parity} is a GPM for $\mathcal{C}^\perp$.\end{Thm}

\begin{exmp}
We continue with the $(2,1)$-MT code $\mathcal{C}$ discussed in Examples \ref{ex_Sh}, \ref{ex_Sh2}, and \ref{ex_Sh3}. From Theorem \ref{MT_dual}, $\mathcal{C}^\perp$ is $(2,1)$-MT over $\mathbb{F}_3$ of length $60$ and dimension $54$. A GPM for $\mathcal{C}^\perp$ can be obtained from Definition \ref{def_parity} and Theorem \ref{MT_Parity_H} as follows:
\begin{equation*}
\mathbf{A}=\begin{pmatrix}
2+2 x + x^4 + x^5+x^6 & \quad 2x(1+x)^4\\
0 & 1
\end{pmatrix}.
\end{equation*}
\begin{equation*}
\begin{split}
\mathbf{A}\left( \frac{1}{x}\right)&=\begin{pmatrix}
x^{-6} + x^{-5} + x^{-4} +2 x^{-1} +2 & \quad 2x^{-1}(x^{-1}+1)^4\\
0 & 1
\end{pmatrix}.\\
\mathbf{A}^*&=\begin{pmatrix}
x^{20-14}\left(x^{-6} + x^{-5} + x^{-4} +2 x^{-1} +2\right) & \quad 2x^{20-40}x^{-1}(x^{-1}+1)^4\\
0 & x^{40-40}
\end{pmatrix}\\
&=\begin{pmatrix}
1 + x + x^{2} +2 x^{5} +2x^6 & \quad 2x^{-25}+ 2x^{-24}+ 2x^{-22}+2x^{-21}\\
0 & 1
\end{pmatrix}.\\
\mathbf{A}^{**}&=\begin{pmatrix}
1 + x + x^{2} +2 x^{5} +2x^6 & \quad x^{-5}+ x^{-4}+ x^{-2}+x^{-1}\\
0 & 1
\end{pmatrix}\\
&=\begin{pmatrix}
1 + x + x^{2} +2 x^{5} +2x^6 & \quad 2x^{15} + 2x^{16} + 2x^{18} + 2x^{19}\\
0 & 1
\end{pmatrix}.\\
\mathbf{H}&=\begin{pmatrix}
1 + x + x^{2} +2 x^{5} +2x^6 &  0 \\
2x^{15} + 2x^{16} + 2x^{18} + 2x^{19}& 1
\end{pmatrix}.
\end{split}
\end{equation*}
The reduced GPM $\mathbf{G}^\perp$ of $\mathcal{C}^\perp$ is obtained by reducing $\mathbf{H}$ to its Hermite normal form. Namely,
\begin{equation*}
\mathbf{G}^\perp=\begin{pmatrix}
1 & 2x+2x^2+x^3+x^4+x^5  \\
 0 &    2+2x+2x^2+x^5+x^6
\end{pmatrix}.
\end{equation*}\end{exmp}

\begin{Coroll}
\label{H_of_QC}
Let $\mathcal{C}$ be an $\ell$-QC code over $\mathbb{F}_q$ of co-index $m$ and the reduced GPM $\mathbf{G}=[g_{i,j}]$. If $\mathbf{A}=[a_{i,j}]$ is the matrix satisfying the identical equation of $\mathbf{G}$, then $\mathcal{C}^\perp$ is $\ell$-QC of co-index $m$ and a GPM 
\begin{equation*}
\mathbf{H}=\left( \mathbf{A}\left(\frac{1}{x}\right) \mathrm{diag}\left[x^{m-d_1},\ldots,x^{m-d_\ell}\right]\pmod{x^m-1}\right)^t,
\end{equation*}
where $d_j=\deg(g_{j,j})$ for $1\le j\le \ell$. The reduction modulo $x^m-1$ is applied to remove negative exponents of $x$ by replacing $x^{-\mu}$ with $x^{m-\mu}$ for $\mu\ge 1$.\end{Coroll}

\begin{Coroll}
\label{H_of_QT}
Let $\mathcal{C}$ be an $(\ell,\lambda)$-QT code over $\mathbb{F}_q$ of co-index $m$ and the reduced GPM $\mathbf{G}=[g_{i,j}]$. If $\mathbf{A}=[a_{i,j}]$ is the matrix satisfying the identical equation of $\mathbf{G}$, then $\mathcal{C}^\perp$ is $(\ell,\frac{1}{\lambda})$-QT of co-index $m$ and a GPM 
\begin{equation*}
\mathbf{H}=\left( \mathbf{A}\left(\frac{1}{x}\right) \mathrm{diag}\left[x^{m-d_1},\ldots,x^{m-d_\ell}\right]\pmod{x^m-\frac{1}{\lambda}}\right)^t,
\end{equation*}
where $d_j=\deg(g_{j,j})$ for $1\le j\le \ell$. The reduction modulo $x^m-\frac{1}{\lambda}$ is applied to remove negative exponents of $x$ by replacing $x^{-\mu}$ with $\lambda x^{m-\mu}$ for $\mu\ge 1$.\end{Coroll}

\begin{Coroll}
\label{H_of_GQC}
Let $\mathcal{C}$ be an $\ell$-GQC code over $\mathbb{F}_q$ of block lengths $(m_1,m_2,\ldots,m_\ell)$ and the reduced GPM $\mathbf{G}=[g_{i,j}]$. If $\mathbf{A}=[a_{i,j}]$ is the matrix satisfying the identical equation of $\mathbf{G}$, then $\mathcal{C}^\perp$ is $\ell$-GQC of block lengths $(m_1,m_2,\ldots,m_\ell)$ and a GPM $\mathbf{H}$, where
\begin{equation*}
\mathrm{Column}_j\left(\mathbf{H}\right)=\mathrm{row}_j \left( \mathbf{A}\left(\frac{1}{x}\right) \mathrm{diag}\left[x^{m_j-d_1},\ldots,x^{m_j-d_\ell}\right]\pmod{x^{m_j}-1}\right)
\end{equation*}
and $d_j=\deg(g_{j,j})$ for $1\le j\le \ell$. The reduction modulo $x^{m_j}-1$ is applied to remove negative exponents of $x$ by replacing $x^{-\mu}$ with $x^{m_j-\mu}$ for $\mu\ge 1$.\end{Coroll}

\begin{Coroll}
\label{S_Dual_Orthog}
Let $\mathcal{C}$ denote an $\ell$-QC, $(\ell,\lambda)$-QT, $\ell$-GQC, or $\left(\lambda_1,\lambda_2,\ldots,\lambda_\ell \right)$-MT code over $\mathbb{F}_q$ and let $\mathbf{H}$ be an $\ell\times\ell$ GPM of $\mathcal{C}^\perp$, e.g., $\mathbf{H}$ is the polynomial matrix defined in Corollary \ref{H_of_QC}, Corollary \ref{H_of_QT}, Corollary \ref{H_of_GQC}, or Definition \ref{def_parity} respectively. If $\mathbf{G}$ is an $\ell\times\ell$ GPM of $\mathcal{C}$, then $\mathcal{C}$ is self-orthogonal if and only if there is a polynomial matrix $\mathbf{U}$ such that $\mathbf{G}=\mathbf{U}\mathbf{H}$. Additionally, $\mathcal{C}$ is self-dual if and only if $\mathbf{U}$ is invertible.\end{Coroll}
\begin{proof}
{Theorem \ref{MT_Parity_H} shows that $\mathbf{H}$ is a GPM for $\mathcal{C}^\perp$. The result then follows from Theorem \ref{Containment}.}
\end{proof}

\begin{exmp}
\label{Ex_on_QC_Optimal}
We investigate in detail one of the QC codes given in Table \ref{tab1}. Let $\mathcal{C}$ be the binary $5$-QC code of length $25$ that has the reduced GPM
\begin{equation*}
\mathbf{G}=\begin{pmatrix}
1+x & 0 & 0 & x+x^4 & x+x^2+x^3+x^4\\
0 & 1+x & 0 & x+x^2+x^3+x^4 & x+x^4\\
0 & 0 & 1+x^5 & 0 & 0\\
0 & 0 & 0 & 1+x^5 & 0\\
0 & 0 & 0 & 0 & 1+x^5\\
\end{pmatrix}.
\end{equation*}
Using brute force method, the weight enumerator of $\mathcal{C}$ is $W_\mathcal{C}(x,y)=x^{25} + 130 x^{17} y^8 + 120 x^{13} y^{12} + 5 x^{9} y^{16}$ and, hence, $d_{\mathrm{min}}\left(\mathcal{C}\right)=8$. By Theorem \ref{dim_MT}, the dimension of $\mathcal{C}$ is $8$. Hence $\mathcal{C}$ is optimal \cite{Grassl}. The identical equation of $\mathbf{G}$ is satisfied by the polynomial matrix
\begin{equation*}
\mathbf{A}=\begin{pmatrix}
1+x+x^2+x^3+x^4&0&0&x+x^2+x^3&x+x^3\\
0& 1+x+x^2+x^3+x^4& 0& x+x^3 &x+x^2+x^3\\
0&0&1&0&0\\
0&0&0&1&0\\
0&0&0&0&1
\end{pmatrix}.
\end{equation*}
Corollary \ref{H_of_QC} provides a GPM for $\mathcal{C}^\perp$ which is given by
\begin{equation*}
\mathbf{H}=\begin{pmatrix}
1+x+x^2+x^3+x^4&0& 0& 0& 0\\
0& 1+x+x^2+x^3+x^4& 0& 0& 0\\
0&0&1&0&0\\
x^2+x^3+x^4&x^2+x^4&0&1&0\\
x^2+x^4 &x^2+x^3+x^4& 0&0&1\\
\end{pmatrix}.
\end{equation*}
Finding the Hermite normal form of $\mathbf{H}$ yields the reduced GPM
\begin{equation*}
\mathbf{G}^\perp=\begin{pmatrix}
1&0&0&x+x^2+x^3&x+x^3\\
0&1&0&x+x^3&x+x^2+x^3\\
0&0&1&0&0\\
0&0&0&1+x+x^2+x^3+x^4&0\\
0&0&0&0&1+x+x^2+x^3+x^4\\
\end{pmatrix}.
\end{equation*}
The dual code $\mathcal{C}^\perp$ is $5$-QC over $\mathbb{F}_2$ of length $25$ and dimension $17$. From the MacWilliams identity, the weight enumerator of $\mathcal{C}^\perp$ is 
\begin{align*}
W_{\mathcal{C}^\perp}(x,y) & = \frac{1}{256}W_\mathcal{C}(x+y,x-y)\\
& = x^{25} + 5 x^{24} y + 10 x^{23} y^2 + 10 x^{22} y^3 +\cdots + 10 x^2 y^{23} + 5 x y^{24} + y^{25}.
\end{align*}
From Corollary \ref{S_Dual_Orthog}, $\mathcal{C}$ is self-orthogonal because $\mathbf{G}=\mathbf{U}\mathbf{G}^\perp$, where
\begin{equation*}
\mathbf{U}=\begin{pmatrix}
1+x&0&0&0&0\\
0&1+x&0&0&0\\
0&0&1+x^5&0&0\\
0&0&0&1+x&0\\
0&0&0&0&1+x\\
\end{pmatrix}.
\end{equation*}\end{exmp}

\section{Combining properties of QC codes}
Suppose that $\mathcal{C}$ is a linear code over $\mathbb{F}_q$ of length $n$. For any codeword $\mathbf{c}\in\mathcal{C}$, the reverse of $\mathbf{c}$ is the vector $\mathbf{r}\in\mathbb{F}_q^n$ obtained by reversing the coordinates of $\mathbf{c}$. That is, $\mathbf{r}=\left(c_{n-1},\ldots,c_{1},c_{0}\right)$ whenever $\mathbf{c}=\left(c_0,c_1,\ldots,c_{n-1} \right)$. The set of all reverses of the codewords of $\mathcal{C}$ forms a linear code over $\mathbb{F}_q$ of length $n$. In this section, $\mathcal{C}$ is restricted to be $\ell$-QC of co-index $m$ and length $n=m\ell$. Hence, the reverse of the codeword given by \eqref{shift2} is
\begin{equation*}
\mathbf{r}\!=\!\left( c_{m-1,\ell}, \ldots, c_{m-1, 2}, c_{m-1,1}, c_{m-2,\ell}, \ldots , c_{m-2,2}, c_{m-2,1}, \ldots, c_{0,\ell}, \ldots, c_{0,2}, c_{0,1} \right)\!.
\end{equation*}

\begin{Def}
\label{def_reversed}
Let $\mathcal{C}$ be an $\ell$-QC code over $\mathbb{F}_q$. The reversed code $\mathcal{R}$ of $\mathcal{C}$ is the set containing the reverse of each codeword of $\mathcal{C}$. That is, 
\begin{equation*}
\mathcal{R}=\left\{\text{The reverse of } \mathbf{c}\  \forall \   \mathbf{c}\in\mathcal{C}\right\}.
\end{equation*}
In addition, $\mathcal{C}$ is called reversible if $\mathcal{C}=\mathcal{R}$, i.e., the reverse of each codeword is a codeword.\end{Def}

\begin{Thm}
Let $\mathcal{C}$ and $\mathcal{R}$ be as in Definition \ref{def_reversed}. Then $\mathcal{R}$ is linear, QC, and has the same index, co-index, dimension, and minimum distance as $\mathcal{C}$.\end{Thm}
\begin{proof}
{In fact, $\mathcal{C}$ and $\mathcal{R}$ are permutation equivalent. Hence, $\mathcal{R}$ is linear and has the same length, dimension and minimum distance as $\mathcal{C}$. Let $\ell$ and $m$ be the index and co-index of $\mathcal{C}$, respectively. For any $\mathbf{r}\in\mathcal{R}$, there exists a codeword $\mathbf{c}\in\mathcal{C}$ such that $\mathbf{r}$ is the reverse of $\mathbf{c}$. It can be seen that $T_\ell\left(\mathbf{r} \right)$ is the reverse of $T_\ell^{-1}\left(\mathbf{c}\right)\in\mathcal{C}$, where $T_\ell$ is the automorphism of $\mathbb{F}_q^n$ that corresponds to $\ell$ coordinates shift. Thus, $T_\ell\left(\mathbf{r} \right)\in\mathcal{R}$ and $\mathcal{R}$ is $\ell$-QC of co-index $m$.}
\end{proof}

The polynomial representation of the codewords of $\mathcal{R}$ is related to that of the codewords of $\mathcal{C}$. Specifically, if we write the polynomial representation of a codeword of $\mathcal{C}$ as $\mathbf{c}=\left(c_1(x), c_2(x),\ldots,c_\ell(x)\right)$, the polynomial representation of its reverse is
\begin{equation}
\label{reverse}
\mathbf{r}=x^{m-1}\left(c_\ell\left(\frac{1}{x}\right),c_{\ell-1}\left(\frac{1}{x}\right),\ldots,c_2\left(\frac{1}{x}\right),c_1\left(\frac{1}{x}\right)\right).
\end{equation}

The following result provides an explicit GPM formula for $\mathcal{R}$, and the proof can be found in \cite{ELDIN2021}.
\begin{Thm}
\label{main_theorem}
Let $\mathcal{C}$ be an $\ell$-QC code over $\mathbb{F}_q$ of co-index $m$ with a GPM $\mathbf{G}=[g_{i,j}]$. Let $\mathbf{J}=\left[\delta_{i,\ell+1-j}\right]$ be the $\ell\times\ell$ backward identity matrix. A GPM for the reversed code $\mathcal{R}$ of $\mathcal{C}$ is given by
\begin{equation}
\label{Intheorem}
\mathbf{F}=\left(\mathrm{diag}\left[x^{m+d_{1}},\ldots,x^{m+d_{\ell}}\right] \mathbf{G}\left(\frac{1}{x}\right)+(1-x^m)\mathrm{diag}\left[g_{1,1}^*,\ldots,g_{\ell,\ell}^*\right]\right) \mathbf{J},
\end{equation}
where $d_i=\deg(g_{i,i})$ and $g_{i,i}^*=x^{d_i}g_{i,i}\left(\frac{1}{x}\right)$ for $1\le i\le \ell$.\end{Thm}

From Theorems \ref{Containment} and \ref{main_theorem}, a QC code with a GPM $\mathbf{G}$ is reversible if and only if there is an invertible polynomial matrix $\mathbf{U}$ such that $\mathbf{F}=\mathbf{U}\mathbf{G}$, where $\mathbf{F}$ is given by \eqref{Intheorem}. By the uniqueness of the Hermite normal form, an equivalent condition for the reversibility of a QC code is that writing $\mathbf{F}$ in its Hermite normal form yields the reduced GPM of $\mathcal{C}$.

\begin{exmp}
\label{Ex_on_QC_Optimal2}
We continue with the binary $5$-QC code $\mathcal{C}$ given in Example \ref{Ex_on_QC_Optimal}. From Theorem \ref{main_theorem}, a GPM for the reversed code $\mathcal{R}$ of $\mathcal{C}$ is
\begin{equation*}
\mathbf{F}=\begin{pmatrix}
x^2+x^3+x^4+x^5&x^2+x^5&0&0&1+x\\
x^2+x^5&x^2+x^3+x^4+x^5&0&1+x&0\\
0&0&1+x^5&0&0\\
0&1+x^5&0&0&0\\
1+x^5&0&0&0&0\\
\end{pmatrix}.
\end{equation*}
One can show that $\mathcal{C}$ is reversible because $\mathbf{F}=\mathbf{U}\mathbf{G}$ for the invertible polynomial matrix
\begin{equation*}
\mathbf{U}=\begin{pmatrix}
x^2+x^4&x^2+x^3+x^4&0&0&1+x\\
x^2+x^3+x^4&x^2+x^4&0&1+x&0\\
0&0&1&0&0\\
0&1+x+x^2+x^3+x^4&0&x+x^3&x+x^2+x^3\\
1+x+x^2+x^3+x^4&0&0&x+x^2+x^3&x+x^3\\
\end{pmatrix}.
\end{equation*}
Equivalently, $\mathcal{C}$ is reversible because the Hermite normal form of $\mathbf{F}$ is $\mathbf{G}$. Therefore, we have shown
that $\mathcal{C}$ is binary optimal self-orthogonal reversible QC code.\end{exmp}

In \cite{ELDIN2021}, computer search is used to present some binary optimal self-orthogonal reversible QC codes with different index values. The minimum distances of these codes are calculated using brute force or the method mentioned in \cite{Grassl2}. We present these computer search results in Table \ref{tab1}. The reader can check the properties of these codes in the same way we used in Examples \ref{Ex_on_QC_Optimal} and \ref{Ex_on_QC_Optimal2}, see also Example \ref{Ex_on_QC_Optimal3} below. Table \ref{tab1} records the non-zero entries $g_{i,j}$ of the reduced GPMs $\mathbf{G}=[g_{i,j}]$ of these codes. In this table, $\{0,6,7,8,10,11\}$ is used to abbreviate the polynomial $1+x^6+x^7+x^8+x^{10}+x^{11}\in\mathbb{F}_2[x]$. We conclude this chapter by summarizing some sufficient and necessary conditions that represent the different states between QC code, its dual code, and its reversed code. We omit the proof of the following result which can be found in \cite{ELDIN2021}.

\begin{table*}[h]
\caption{Binary optimal self-orthogonal reversible QC codes}
\label{tab1}
\begin{tabular}{|c|c|c|c|c|}
\hline&&&&\\[-.9em]
$\ell$ & $n$ & $k$ & $d_\mathrm{min}$ & $G=(g_{i,j})$  \\[.2em] \hline&&&&\\[-.9em]

\multirow{2}{*}{2} & \multirow{2}{*}{64} & \multirow{2}{*}{32} & \multirow{2}{*}{12} & $g_{1,1}=\{0\}$,\  $g_{2,2}=\{0,32\}$, \\
&  &  &  & $g_{1,2}=\{2,5,6,7,8,9,10,11,12,15,16,18,19,20,22,24,25,28,29,30,31\}$\\
\hline&&&&\\[-.9em]
\multirow{2}{*}{3} & \multirow{2}{*}{36} & \multirow{2}{*}{6} & \multirow{2}{*}{16} & $g_{1,1}=\{0,1,2,4,5,6\}$,\  $g_{1,2}=\{1,5,7,11\}$,\\  
&  &  &  & $g_{1,3}=\{0,6,7,8,10,11\}$,\  $g_{2,2}=g_{3,3}=\{0,12\}$ \\
\hline&&&&\\[-.9em]

\multirow{4}{*}{4} & \multirow{4}{*}{68} & \multirow{4}{*}{34} & \multirow{4}{*}{12} &
$g_{1,1}=g_{1,2}=\{ 0\}$,\  $g_{1,3}=\{0,3,4,7,10,11,14\}$, \\
&  &  &  & $g_{1,4}=\{1,2,6,7,10,12,14\}$, \ $g_{2,2}=\{0,1\}$, \\
&  &  &  & $g_{2,3}=\{1,4,5,9,10,15\}$,\  $g_{2,4}=\{0,3,4,7,8,9,12,14\}$, \  $g_{4,4}=\{0,17\}$,\\
&  &  &  & $g_{3,3}=g_{3,4}=\{0, 1, 2, 3, 4, 5, 6, 7, 8, 9,10,11,12,13,14,15,16\}$ \\\hline&&&&\\[-.9em]

\multirow{2}{*}{5} & \multirow{2}{*}{25} & \multirow{2}{*}{8} & \multirow{2}{*}{8} &
$g_{1,1}=g_{2,2}=\{0, 1\}$,\  $g_{1,4}=g_{2,5}=\{1,4\}$, \\
&  &  &  &  $g_{1,5}=g_{2,4}=\{1,2,3,4\}$,\  $g_{3,3}=g_{4,4}=g_{5,5}=\{0,5\}$ \\\hline&&&&\\[-.9em]

\multirow{4}{*}{6} & \multirow{4}{*}{36} & \multirow{4}{*}{18} & \multirow{4}{*}{8} &
$g_{1,1}=g_{1,3}=g_{2,2}=g_{2,5}=\{ 0\}$,\  $g_{1,4}=\{2,4\}$,\\
&  &  &  &  $g_{1,5}=g_{4,4}=g_{4,6}=\{0,1,2,3,4,5\}$, \ $g_{1,6}=\{0,1,3,5\}$,\\
&  &  &  &  $g_{2,4}=\{3\}$,\  $g_{2,6}=\{0,1,2,4,5\}$, \ $g_{3,3}=\{0,1\}$,\  $g_{3,4}=\{0,3,4\}$,\\
&  &  &  &  $g_{3,5}=\{3,4\}$,\  $g_{3,6}=\{0,2,5\}$,\  $g_{5,5}=g_{6,6}=\{0,6\}$ \\\hline&&&&\\[-.9em]

\multirow{4}{*}{7} & \multirow{4}{*}{42} & \multirow{4}{*}{14} & \multirow{4}{*}{12} &
$g_{1,1}=\{0\}$,\  $g_{1,3}=\{0,1,2,3\}$,\  $g_{1,4}=\{0,3\}$,\  $g_{1,5}=\{5\}$,\  $g_{1,6}=\{2,3,4,5\}$, \\
&  &  &  & $g_{2,2}=\{0,1\}$,\  $g_{2,3}=\{2\}$,\  $g_{2,4}=\{1,4\}$,\  $g_{2,5}=\{0,1,4,5\}$, \\
&  &  &  & $g_{2,6}=g_{3,6}=g_{4,4}=g_{4,6}=\{0,1,2,3,4,5\}$,\  $g_{2,7}=\{1\}$,\\
&  &  &  & $g_{3,3}=\{0,2,4\}$,\   $g_{3,7}=\{1,3,5\}$, \  $g_{5,5}=g_{6,6}=g_{7,7}=\{0,6\}$ \\\hline&&&&\\[-.9em]

\multirow{4}{*}{8} & \multirow{4}{*}{40} & \multirow{4}{*}{20} & \multirow{4}{*}{8} &
$g_{1,1}=g_{2,2}=g_{3,3}=g_{4,4}=\{0\}$,\  $g_{1,5}=g_{4,8}=\{0,2,4\}$,\\
&  &  &  & $g_{1,6}=g_{2,5}=g_{3,8}=g_{4,7}=\{0,1,4\}$, \ $g_{1,7}=g_{2,8}=\{0,2\}$,\  $g_{1,8}=\{1,2,4\}$, \\
&  &  &  & $g_{2,6}=g_{3,7}=\{3\}$,\  $g_{2,7}=\{2\}$,  \ $g_{3,5}=g_{4,6}=\{1\}$,\  $g_{3,6}=\{1,4\}$,\\
&  &  &  & $g_{4,5}=\{1,2,3,4\}$, \  $g_{5,5}=g_{6,6}=g_{7,7}=g_{8,8}=\{0,5\}$ \\\hline&&&&\\[-.9em]

\multirow{6}{*}{9} & \multirow{6}{*}{54} & \multirow{6}{*}{24} & \multirow{6}{*}{12} &
$g_{1,1}=g_{1,2}=g_{3,3}=g_{3,4}=\{0 \}$,\  $g_{1,6}=\{ 1\}$,\  $g_{1,7}=\{ 1,3,5\}$,\\
&  &  &  & $g_{1,8}=\{ 0,2\}$, \ $g_{1,9}=g_{2,5}=g_{3,5}=g_{4,5}=\{ 1,2,4,5\}$, \  $g_{2,2}=g_{4,4}=\{ 0,1\}$,\\
&  &  &  & $g_{2,6}=g_{4,8}=\{ 0,1,4\}$, \ $g_{2,7}=\{ 0,1,2,3,4\}$,\  $g_{2,8}=\{ 1,4\}$,\  $g_{2,9}=\{ 0,2,3,4\}$, \\
&  &  &  & $g_{3,6}=\{ 3,4\}$,\  $g_{3,7}=\{ 0,1,3,5\}$,\ $g_{3,8}=\{ 2\}$,\  $g_{3,9}=\{ 2,3,5\}$, \\
&  &  &  & $g_{4,6}=\{ 0,1,2,3\}$, \ $g_{4,7}=\{ 0,1,2,5\}$,\  $g_{4,9}=\{ 0,2,4\}$, \\
&  &  &  & $g_{5,5}=g_{7,7}=g_{9,9}=\{ 0,6\}$, \ $g_{6,6}=g_{6,7}=g_{8,8}=g_{8,9}=\{ 0,1,2,3,4,5\}$\\\hline&&&&\\[-.9em]

\multirow{5}{*}{10} & \multirow{5}{*}{40} & \multirow{5}{*}{20} & \multirow{5}{*}{8} &
$g_{1,1}=g_{2,2}=g_{3,3}=g_{4,4}=g_{5,5}=g_{3,6}=g_{5,8}=\{0\}$,\\
&  &  &  &  $g_{1,6}=g_{3,8}=g_{5,10}=\{ 0,1\}$,\ $g_{1,7}=g_{1,9}=g_{2,10}=g_{4,10}=g_{5,6}=\{ 2,3\}$,\\
&  &  &  &  $g_{1,8}=g_{2,7}=g_{3,10}=g_{4,9}=\{ 1,2\}$,\ $g_{1,10}=g_{4,6}=g_{5,7}=\{ 3\}$, \\
&  &  &  &  $g_{2,6}=g_{5,9}=\{ 0,1,2\}$,\  $g_{2,8}=g_{3,9}=\{ 1\}$,\  $g_{2,9}=g_{4,7}=\{ 2\}$,\\
&  &  &  &  $g_{3,7}=g_{4,8}=\{ 0,2,3\}$ ,\  $g_{6,6}=g_{7,7}=g_{8,8}=g_{9,9}=g_{10,10}=\{ 0,4\}$\\\hline
\end{tabular}
\end{table*}

\begin{Thm}
\label{concluding_Th}
Let $\mathcal{C}$ be an $\ell$-QC code over $\mathbb{F}_q$ of co-index $m$. Let $\mathbf{G}$ be the reduced GPM of $\mathcal{C}$ and let $\mathbf{A}$ be the matrix that satisfies the identical equation of $\mathbf{G}$. Let $\mathcal{C}^\perp$ and $\mathcal{R}$ be the dual and reversed code of $\mathcal{C}$, respectively. Then,
\begin{enumerate}
\item $\mathcal{C}^\perp \supseteq \mathcal{R}$ if and only if $\mathbf{G}\mathbf{J}\mathbf{G}^t \equiv \mathbf{0}_{\ell\times\ell} \pmod{x^m-1}$. Therefore, if $\mathcal{C}$ is reversible, then $\mathcal{C}$ is self-orthogonal if and only if $\mathbf{G}\mathbf{J}\mathbf{G}^t\equiv \mathbf{0}_{\ell\times\ell} \pmod{x^m-1}$.
\item $\mathcal{C}^\perp \subseteq \mathcal{R}$ if and only if $\mathbf{A}^t\mathbf{J}\mathbf{A} \equiv \mathbf{0}_{\ell\times\ell} \pmod{x^m-1}$. Therefore, if $\mathcal{C}$ is reversible, then $\mathcal{C}^\perp \subseteq \mathcal{C}$ if and only if $\mathbf{A}^t\mathbf{J}\mathbf{A}\equiv \mathbf{0}_{\ell\times\ell} \pmod{x^m-1}$.
\item If $\mathcal{C}$ is reversible, then $\mathcal{C}$ is self-dual if and only if $\mathbf{G}\mathbf{J}\mathbf{G}^t\equiv\mathbf{A}^t\mathbf{J}\mathbf{A}\equiv \mathbf{0}_{\ell\times\ell} \pmod{x^m-1}$ if and only if $\mathbf{A}=\mathbf{J}\mathbf{G}^t\mathbf{J}$.
\item If $\mathbf{A}=\mathbf{J}\mathbf{G}^t\mathbf{J}$, then $\mathcal{C}$ is self-dual if and only if $\mathcal{C}$ is reversible,
\item If $\mathcal{C}$ is self-dual, then $\mathcal{C}$ is reversible if and only if $\mathbf{A}=\mathbf{J}\mathbf{G}^t\mathbf{J}$.

\end{enumerate}\end{Thm}

\begin{exmp}
\label{Ex_on_QC_Optimal3}
We continue with the binary optimal self-orthogonal reversible QC code discussed in Example \ref{Ex_on_QC_Optimal2}. Since $\mathcal{C}$ is reversible, Theorem \ref{concluding_Th} can be used to check the self-orthogonality of $\mathcal{C}$, see Example \ref{Ex_on_QC_Optimal}. Specifically, $\mathcal{C}$ is self-orthogonal if and only if $\mathbf{G}\mathbf{J}\mathbf{G}^t\equiv \mathbf{0}_{\ell\times\ell} \pmod{1+x^5}$. In fact, we have
\begin{equation*}
\mathbf{G}\mathbf{J}\mathbf{G}^t=(1+x^5)\begin{pmatrix}
0&0&0&0&1+x\\
0&0&0&1+x&0\\
0&0&1+x^5&0&0\\
0&1+x&0&0&0\\
1+x&0&0&0&0
\end{pmatrix}.
\end{equation*}\end{exmp}

\bibliographystyle{spmpsci}
\bibliography{References}

\end{document}